\documentclass[onecolumn]{IEEEtran}

\usepackage{amsfonts}

\usepackage{times,amssymb,amsthm,amsmath,amsfonts,bbm,mathrsfs,color,subfigure,graphicx,enumitem,booktabs,cite,euscript}
\usepackage[normalem]{ulem}
\usepackage[font=footnotesize]{caption}
\allowdisplaybreaks
\usepackage[unicode]{hyperref}
\hypersetup{
    colorlinks,
    linkcolor={blue!80!black},
    citecolor={blue!80!black},
    urlcolor={blue!80!black}
}

\newtheorem{theorem}{Theorem}[section]
\newtheorem{lemma}[theorem]{Lemma}

\newtheorem{proposition}[theorem]{Proposition}

\newtheorem{corollary}[theorem]{Corollary}

\newtheorem{example}{Example}[section]
\usepackage[top=1.2in,bottom=1.2in,left=1.2in,right=1.2in]{geometry}

\newcommand\nc\newcommand
\nc\ffa{{\boldsymbol a}}\nc\ffA{{\boldsymbol A}}\nc\cA{{\EuScript A}}
\nc\ffb{{\boldsymbol b}}\nc\ffB{{\boldsymbol B}}\nc\cB{{\EuScript B}}
\nc\ffc{{\boldsymbol c}}\nc\ffC{{\boldsymbol C}}\nc\cC{{\mathscr C}}
\nc\ffd{{\boldsymbol d}}\nc\ffD{{\boldsymbol D}}\nc\cD{{\EuScript D}}
\nc\ffe{{\boldsymbol e}}\nc\ffE{{\boldsymbol E}}\nc\cE{{\EuScript E}}
\nc\fff{{\boldsymbol f}}\nc\ffF{{\boldsymbol F}}\nc\cF{{\mathscr F}}
\nc\ffg{{\boldsymbol g}}\nc\ffG{{\boldsymbol G}}\nc\cG{{\EuScript G}}
\nc\ffh{{\boldsymbol h}}\nc\ffH{{\boldsymbol H}}\nc\cH{{\EuScript H}}
\nc\ffi{{\boldsymbol i}}\nc\ffI{{\boldsymbol I}}\nc\cI{{\mathcal I}}
\nc\ffj{{\boldsymbol j}}\nc\ffJ{{\boldsymbol J}}\nc\cJ{{\EuScript J}}
\nc\ffk{{\boldsymbol k}}\nc\ffK{{\boldsymbol K}}\nc\cK{{\EuScript K}}
\nc\ffl{{\boldsymbol l}}\nc\ffL{{\boldsymbol L}}\nc\cL{{\EuScript L}}
\nc\ffm{{\boldsymbol m}}\nc\ffM{{\boldsymbol M}}\nc\cM{{\EuScript M}}
\nc\ffn{{\boldsymbol n}}\nc\ffN{{\boldsymbol N}}\nc\cN{{\EuScript N}}
\nc\ffo{{\boldsymbol o}}\nc\ffO{{\boldsymbol O}}\nc\cO{{\EuScript O}}
\nc\ffp{{\boldsymbol p}}\nc\ffP{{\boldsymbol P}}\nc\cP{{\EuScript P}}
\nc\ffq{{\boldsymbol q}}\nc\ffQ{{\boldsymbol Q}}\nc\cQ{{\EuScript Q}}
\nc\ffr{{\boldsymbol r}}\nc\ffR{{\boldsymbol R}}\nc\cR{{\EuScript R}}
\nc\ffs{{\boldsymbol s}}\nc\ffS{{\boldsymbol S}}\nc\cS{{\EuScript S}}
\nc\fft{{\boldsymbol t}}\nc\ffT{{\boldsymbol T}}\nc\cT{{\EuScript T}}
\nc\ffu{{\boldsymbol u}}\nc\ffU{{\boldsymbol U}}\nc\cU{{\EuScript U}}
\nc\ffv{{\boldsymbol v}}\nc\ffV{{\boldsymbol V}}\nc\cV{{\mathscr V}}
\nc\ffw{{\boldsymbol w}}\nc\ffW{{\boldsymbol W}}\nc\cW{{\mathscr W}}
\nc\ffx{{\boldsymbol x}}\nc\ffX{{\boldsymbol X}}\nc\cX{{\EuScript X}}
\nc\ffy{{\boldsymbol y}}\nc\ffY{{\boldsymbol Y}}\nc\cY{{\mathscr Y}}
\nc\ffz{{\boldsymbol z}}\nc\ffZ{{\boldsymbol Z}}\nc\cZ{{\EuScript Z}}
\nc{\bb}{{\mathbbm{1}}}
\nc\reals{{\mathbb R}}
\nc{\ff}{{\mathbb F}}
\nc{\PP}{{\mathbb P}}

\DeclareMathOperator{\sgn}{sgn}

\usepackage{tikz,tikz-cd}
\usetikzlibrary{positioning}
\usetikzlibrary{calc}

\usetikzlibrary{matrix,arrows,decorations.pathmorphing}

\DeclareSymbolFont{bbold}{U}{bbold}{m}{n}
\DeclareSymbolFontAlphabet{\mathbbold}{bbold}

\usepackage{xargs}
\usepackage[colorinlistoftodos,prependcaption,textwidth=.7in]{todonotes}
\newcommandx{\yellownote}[2][1=]{\todo[linecolor=yellow,backgroundcolor=yellow!25,bordercolor=yellow,#1]{#2}}
\newcommandx{\greennote}[2][1=]{\todo[inline,linecolor=olive,backgroundcolor=green!25,bordercolor=olive,#1]{#2}}
\usepackage[normalem]{ulem}
\newcommand\redout{\bgroup\markoverwith{\textcolor{red}{\rule[0.5ex]{2pt}{0.8pt}}}\ULon}

\begin{document}
	
\title{Node repair on connected graphs, Part II}
	\author{\IEEEauthorblockN{Adway Patra} \hspace*{1in}
		\and \IEEEauthorblockN{Alexander Barg}}
	\date{}
	\maketitle

	\begin{abstract}
We continue our study of regenerating codes in distributed storage systems where connections between the nodes are constrained by
a graph. In this problem, the failed node downloads the information stored at a subset of vertices of the graph for the purpose
of recovering the lost data. This information is moved across the network, and the cost of node repair is determined by the graphical distance from the helper nodes to the failed node. This problem was formulated in our recent work ({\em IEEE IT Transactions}, May 2022) where we showed that processing of the information at the intermediate nodes can yield savings in 
repair bandwidth over the direct forwarding of the data. 

While the previous paper was limited to the MSR case, here we extend our study to the case of general regenerating codes. 
We derive a {\em lower bound} on the repair bandwidth and formulate {\em repair
procedures with intermediate processing} for several families of regenerating codes, with an emphasis on the recent
constructions from multilinear algebra. We also consider the task of data retrieval for codes on graphs, deriving a lower bound on the communication bandwidth and showing that it can be attained at the MBR point of the storage-bandwidth tradeoff curve.
\end{abstract}
	
	\renewcommand{\thefootnote}{\arabic{footnote}}
	\setcounter{footnote}{0}		
		{\renewcommand{\thefootnote}{}\footnotetext{
			\vspace*{-.15in}
	
			\noindent\rule{1.5in}{.4pt}
			
			{An extended abstract of this paper is published in Proceedings of the IEEE International Symposium on Information Theory, Helsinki, Finland, July 2022, pp. 1608--1612.
						
			The authors are with Dept. of ECE and ISR, University of Maryland, College Park, MD 20742. Emails: \{apatra,abarg\}@umd.edu. 
					This research was supported by NSF grants CCF2110113 and CCF2104489.
	}}}
	\renewcommand{\thefootnote}{\arabic{footnote}}
	\setcounter{footnote}{0}
	
	\section{Introduction}\label{Introduction}

A distributed storage system is formed of a number of nodes connected by communication links which carry the information to accomplish the two basic tasks performed in the system, namely data recovery and node repair. The amount of information sent over the links is
a key metric of the system efficiency. The problem of node repair has been widely studied in the literature in the last decade following
its introduction in \cite{Dimakis10}.
The system is modeled as $n$ storage nodes each with capacity of $l$ units, used to store a file $\cF$ of size $M$, such that the following two properties are met:
	\begin{itemize}
		\item{(\em{Data retrieval})} The entire file can be recovered by accessing any $k < n$ nodes. 
		\item{(\em{Repair})} If a single node fails, data from $d$ surviving, or helper, nodes is used to restore the lost data. 
		We assume that each of the helper nodes contributes $\beta\le l$ units of data, and that $k \le d \le n - 1$. The parameter $\beta$ is 
		called the {\em per-node repair bandwidth}.
	\end{itemize}
We write the parameters of a regenerating code as $(n,k,d,\beta,l,M).$ 
The fundamental tradeoff between the file size and the repair bandwidth is expressed by the bound of \cite{Dimakis10} 
which has the form
	\begin{equation}\label{eq:sbt}
	M \le \sum_{i=1}^k\min \{l,(d-i+1)\beta\}.
	\end{equation}
This bound can be attained for the two corner points of the curve \eqref{eq:sbt}, giving rise, respectively, to Minimum Storage Regenerating (MSR) codes and Minimum Bandwidth Regenerating (MBR) codes. The corresponding values of $l$ and $\beta$ are found when the minimum in \eqref{eq:sbt} for all $i$ is attained by the first and the second term, respectively, and have the form
 \begin{equation}\label{eq:MSR-MBR}
  \begin{array}{lcc}
  \text{MSR:}& l=\frac Mk, &\beta=\frac l{d-k+1};\\[.1in]
  \text{MBR:} &\beta=\frac{M}{dk-k(k-1)/2}, &l=d\beta.
   \end{array}
 \end{equation}
 
The repair problem has been studied in two versions, called functional and exact repair. Under exact repair, the contents of the failed
node is recovered in the exact form, while for functional repair the node can be restored to a different value as long as it continues
to support the two properties above. While for functional repair the entire bound \eqref{eq:sbt} is achievable, for the more stringent exact repair requirement there is a gap between the achievable file size and the bound, first demonstrated in \cite{Tian14} in an example and then extended in \cite{Sasidharan2014,Mohajer2015} to all sets of parameters $(n,k,d).$ 
 
The MSR case is the most widely studied in the literature. Several general constructions of MSR codes have been proposed in recent years, among them product matrix codes \cite{Rashmi11} and their generalization in \cite{Duursma2020}, diagonal matrix codes \cite{Ye16a}, and others.
In this work we study MSR codes as well as codes for the interior points of the trade-off curve. Several interior-point code families are known in the literature, among them layered and determinant codes \cite{Senthoor2015,Elyasi2019,Elyasi2020}, and a recent construction of \cite{Duursma2021}, called Moulin codes by its authors. Here we cite only papers that are directly related to our work. Generally, the subject of regenerating codes has accumulated vast literature, and we refer the reader  to the survey by Ramkumar et al. \cite{Ramkumar2022} for a very readable and detailed overview.
	
In this paper we continue the study of regenerating codes on graphs introduced in our earlier work \cite{Patra2021}. This variant of the node repair problem  assumes that communication between the nodes is constrained by a (connected) graph $G(V,E)$ and the cost of sending a unit of information from $v_i$ to $v_j$ is determined by the graphical distance $\rho(v_i,v_j)$ in $G$. Similarly, the data retrieval problem
is bound by the same constraints. Placing the nodes of the system on a graph results in a bias in the information cost of node repair in favor of the helper nodes closer to the failed node $v_f,$ and suggests that the closer nodes combine the information received from the outer extremes of the helper set before relaying it to the failed node. We call this approach {\em Intermediate Processing}, or IP, as opposed to direct relaying. 

Prior to our works, repair on graphs using MSR codes was considered in \cite{GeramiXiao2014,LuXuanFu2014} for particular examples of graphs. A somewhat similar setting arises when it is assumed that transmitting the data from a subset of nodes incurs larger cost than for the remaining nodes \cite{AkhKiaGha2010cost,SohnChoYooMoo2018} or that the links between the nodes (in a fully connected graph) are assigned weights that translate into the cost of sending symbols over them \cite{LiMowDengWu2022}. Our assumptions and results are more general in the sense
that these papers relied on direct relaying only and do not afford the option of incorporating intermediate data processing. Another 
difference arises because the heterogeneity in the network in these works is fixed irrespective of the location of the failed nodes. At the same time, our setting implies that cost of transmission from the node may be high or low depending on whether it is far from the failed node or is among its immediate neighbors. Arguably this accounts for a more uniform treatment of the nodes in the network. 

Another related communication problem is that of network coding \cite{Yeung2006} wherein (in its simplest version) the data is transmitted from
a single fixed source to multiple destinations, and where it is assumed that the intermediate nodes combine the chunks of data on their incoming edges. While intermediate processing is a shared feature between node repair on graphs and network codes, they address different tasks and rely on different kinds of code constructions.

In \cite{Patra2021} we focused on the repair problem for MSR codes, proving a lower bound on the communication complexity (bandwidth) of node repair on graphs. We also showed that linear MSR codes can be modified to implement IP, attaining the complexity lower bound and achieving savings in the repair bandwidth over simple relaying. We refer to the introduction of \cite{Patra2021} for a more detailed
discussion, including the motivation for this problem. Initially the goal of this paper had been to extend the results of \cite{Patra2021} to intermediate
points of the storage-bandwidth curve; however it has become clear that the savings from the IP procedure are related more to the
linearity of the considered codes than to the MSR property. Already in \cite{Patra2021} we have pointed out that IP repair is 
possible for any linear MSR code, although the details of the procedure depend on the family and are not immediate to work out.
Therefore, while in this work we study IP for intermediate-point codes, we again start with the MSR case, notably the
product-matrix codes. In doing so, we shift the perspective, viewing them as {\em evaluation codes}, i.e., codes whose encoding can be phrased as evaluation of a linear functional written in a convenient algebraic form. We rewrite the IP repair procedure
of product-matrix codes from \cite{Patra2021}, which enables us to extend it to a much more general class of codes introduced recently by Duursma and Wang \cite{Duursma2020}. This in turn prepares the way for the analysis of intermediate-point codes, and we begin with implementing IP repair for the Moulin codes of Duursma et al. \cite{Duursma2021} which also fall under the evaluation category. 

To set up a benchmark for IP repair, in Sec.~\ref{sec:BRB} we prove a general lower bound on the repair bandwidth which extends a result of \cite{Patra2021}. In Sec.~\ref{sec:pm} we rephrase the IP repair of product-matrix codes in the format of evaluation codes, and in Sec.\ref{sec:GPM} we formulate the IP repair for the codes of \cite{Duursma2020}. Then in Sec.~\ref{sec:moulin} we turn to interior point codes
of  \cite{Duursma2021}, formulating IP node repair and estimating the repair bandwidth. Further in Sec.~\ref{sec:det}, \ref{sec:cascade} we consider the families of determinant and cascade codes \cite{Elyasi2016,Elyasi2020} observing that their construction makes them a natural candidate for IP repair on graphs. In Sec.~\ref{sec:retrieval} we analyze the problem of data retrieval for codes on graphs, deriving a lower bound for the communication and a matching code construction, which comes from the MBR version of the product matrix codes. We end the paper with two brief sections on node repair with noisy edges and partial node repair. 

To summarize, our main results are related to implementing the IP techniques for several families of interior-point codes as well as MSR codes.  
Note however that, unlike the MSR case, we are not able to bridge the gap between the lower bound on the minimum possible required information transmission and what is achievable using the constructions designed in this work (we do not know whether this is a deficiency of the
bounds or of the constructions). We also formulate and analyze a model of data retrieval for regenerating codes on graphs.
	
\section{Bounds on the Repair Bandwidth}\label{sec:BRB}
For a finite field ${F}=\ff_q$ we consider a code $\cC\subset {F}^{nl}$ whose codewords $(C_i,i=1,\dots,n)$ 
are represented by $l\times n$ matrices over ${F}$.
We assume that each coordinate (a vector in ${F}^l$) is written on a single storage node, and that a failed node amounts to having its 
coordinate erased. Limited connectivity of the network is modeled as placing each node on a vertex of a graph
$G(V,E)$ with $|V|=n,$ where each node has direct access only to its immediate neighbors in $G$. Suppose further that the 
coordinate $C_f$ for some $f\in[n]$ is erased, i.e., that the node $f\in [n]$ has failed. Below we denote the vertex in $V$ that corresponds to $f$ by $v_f$ and use $f$ and $v_f$ interchangeably. Let $D\subset V\backslash\{v_f\}, |D|=d$ be the set of nodes in the graph $G$ that are the closest to $v_f$ in terms of graph distance. This set can be found by running breadth-first search with $v_f$
as the root node. Let $G_{f,D} = (V_{f,D},E_{f,D})$ be the subgraph spanned by $D \cup \{v_f\}$ (an example of this subgraph is shown
in Fig.~\ref{fig:graph1}). To repair the failed node, the helper nodes provide information which is communicated to $v_f$  over the edges in $E_{f,D}$. Each helper node in the graph, starting from the nodes farthest from the failed node, sends its repair data ($\beta$ symbols each) to the next node along the shortest path towards $v_f$. An intermediate node can simply collect this data, supplement it with its own information, and forward it along the path to $v_f$ (Accumulate-and-Forward, or AF). The AF technique can be wasteful in high-depth repair graphs since the same data gets transmitted multiple times. This gives rise to the problem of attaining savings by processing the information in the intermediate
nodes relying on the IP approach, an idea that has already been explored for MSR codes in \cite{Patra2021}.  
\vspace*{.05in}

\begin{figure}[th]\begin{center}\scalebox{0.27}{\begin{tikzpicture}
				[
				vertex_style/.style={circle, draw, fill,minimum size=0.08cm,scale=1.25},
				vertex_style1/.style={circle, draw, fill=blue,minimum size=0.08cm,scale=1.25}
				]
				
				\useasboundingbox (-10,-10) rectangle (10,10);
				
				\begin{scope}[rotate=90]
					
					\node[circle,draw=red,fill=red,minimum size=0.1cm,scale=1.5,label={[xshift=1.4cm, yshift=-1cm,minimum size=0.5cm,color=black,scale=4]$v_f$}] (0) at (canvas polar cs: radius=0cm,angle=0){};
					
					\foreach \x/\y in {36/1,108/2,180/3,252/4,324/5}{
						\node[vertex_style] (\y) at (canvas polar cs: radius=4cm,angle=\x){};
					}
					\foreach \x/\y in {0/1,0/2,0/3,0/4,0/5}{
						\path[-] (\x) edge [black,ultra thick] (\y);
					}
				
					\foreach \x/\y in {13.5/6,28.5/7,43.5/8,58.5/9}{
						\node[vertex_style1] (\y) at (canvas polar cs: radius=8.5cm,angle=\x){};
						\path[-] (1) edge [blue,thick] (\y); 
					}
					\foreach \x/\y in {85.5/10,100.5/11,115.5/12,130.5/13}{
						\node[vertex_style1] (\y) at (canvas polar cs: radius=8.5cm,angle=\x){};
						\path[-] (2) edge [blue,thick] (\y); 
					}
					\foreach \x/\y in {85.5+72/14,100.5+72/15,115.5+72/16,130.5+72/17}{
						\node[vertex_style1] (\y) at (canvas polar cs: radius=8.5cm,angle=\x){};
						\path[-] (3) edge [blue,thick] (\y); 
					}
					\foreach \x/\y in {85.5+144/19,100.5+144/20,115.5+144/21,130.5+144/22}{
						\node[vertex_style1] (\y) at (canvas polar cs: radius=8.5cm,angle=\x){};
						\path[-] (4) edge [blue,thick] (\y); 
					}
					\foreach \x/\y in {85.5+216/23,100.5+216/24,115.5+216/25,130.5+216/26}{
						\node[vertex_style1] (\y) at (canvas polar cs: radius=8.5cm,angle=\x){};
						\path[-] (5) edge [blue,thick] (\y); 
					}
					\draw[black,ultra thick,dashed,label=x] (0,0) circle (4.0cm) ;
					\draw[blue,ultra thick,dashed] (0,0) circle (8.5cm);
					\end{scope}
		\end{tikzpicture}}
		\caption{Node repair on a graph: The failed node $v_f$ and the set $D$ of helper nodes, forming the repair tree $T_f$.}\label{fig:graph1}
	\end{center}
\end{figure}
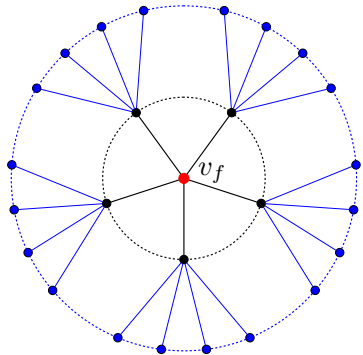

	\vspace*{-.1in}\subsection{Lower bounds on the repair bandwidth} 
In this section we derive a lower bound on the minimum required transmission for a set of helper nodes for repair of the failed node. Suppose that the information stored at the vertices is described by random variables $W_i,i\in[n]$ that have some joint distribution on $({F}^l)^n$ and satisfy $H(W_i)=l$ for all $i$, where $H(\cdot)$ is the entropy.  For a subset $A\subset V$ we write $W_A=\{W_i,i\in A\}.$
Denote by $S_i^f$ the information provided to $v_f$ by the $i$th helper node in the traditional fully connected repair 
scheme, and let $S_D^f=\{S_i^f,i\in D\}.$ By definition we have
   $$
   H(S_i^f) = \beta,\quad H(S_i^f|W_i)=0,\quad H(W_f|S_D^f)=0.
   $$
We also assume that $H(\cF|W_B)=0$ for any $B\subset[n],|B|=k,$ which supports the data retrieval property.
The following result was proved in \cite{Shah2012}:
	\begin{lemma}\label{lemma:claim1} For any $A\subset [n], |A|\le d$ and $i\not\in A$
	  $$
	  H(W_i|W_A)\le \min(l,(d-|A|)\beta).
	  $$
	\end{lemma}
The next lemma forms a simple extension of \cite[Lemma II.1]{Patra2021}, generalizing it to all exact regenerating codes.
	\begin{lemma}\label{lemma:bound} Let $v_f, f\in[n]$ be the failed node.
		For a subset of the helper nodes $E \subset D$ let $R_E^f$ be a function of $S_E^f$ such that 
		\begin{equation}\label{eq:eqtn1}
			H(W_f|R_E^f, S_{D\backslash E}^f) = 0.
		\end{equation} 
		If $|E| \ge d-k+1$, then $H(R_E^f) \ge (d-k+1)\beta.$ In particular, at the MSR point we have $H(R_E^f)\ge l.$
	\end{lemma}
\begin{proof}
\noindent{\em Proof:}
		By the assumption \eqref{eq:eqtn1}, given the contents of all the nodes in $D\backslash E,$ the information contained in $R_E^f$ is sufficient to repair $v_f$, i.e., 
		\begin{equation}\label{eq:eqtn2}
			H(W_f|R_E^f, W_{D\backslash E})=0.
		\end{equation}
		We have $|D\backslash E| \le k-1$. Consider a set $A \subset E$ with $|A| = k-1-|D\backslash E|$. Now, by \eqref{eq:eqtn2}
		\begin{equation}\label{eq:eqtn3}
			H(R_E^f, W_{D\backslash E}, W_{A}) = H(R_E^f, W_{D\backslash E}, W_f, W_{A}) = M,
		\end{equation}
where the first equality in \eqref{eq:eqtn3} follows from \eqref{eq:eqtn2} and the chain rule, and the second follows from reconstruction property because 
		$|D\backslash E|+|A|+1 = k$. Next observe that
		\begin{align}\nonumber
			H(R_E^f, W_{ D\backslash E}, W_{A}) &\le H(R_E^f)+H( W_{D\backslash E}, W_{A}),
		\end{align}
		and so 
\vspace*{-.1in}		\begin{align*}
			H(R_E^f) &\ge M - H( W_{D\backslash E}, W_{A})\\
			&\ge M- \sum_{i=1}^{k-1}\min \{l,(d-i+1)\beta\},
		\end{align*}
where the last inequality follows from Lemma \ref{lemma:claim1}. The largest value of $M$ is given in \eqref{eq:sbt}, implying the
claim of the lemma. 
\end{proof}

Note that at the MSR point $(d-k+1)\beta=l$ and we recover Lemma II.1 from \cite{Patra2021}. In that work we also showed that $H(R_E^f)=l$ is achievable at the MSR point. At the same time for all other points of the tradeoff curve, $(d-k+1)\beta<l.$ Below in this paper we show
that the value $H(R_E^f)=l$ can be achieved by some code families, and hence it might be possible to improve the bound. The following lemma from \cite{Mohajer2015} shows that in certain situations this is indeed the case.
	\begin{lemma}[\!\!\cite{Mohajer2015}, Lemma 2]
		For any pair of disjoint sets $E, B \subseteq D$ with $i \notin E \cup B$ we have 
		$$H(S_E^f|W_B) \ge \frac{|E|}{d-|B|}H(S_D^f|W_B).$$
	\end{lemma}
Taking $B=\emptyset$ and noting that $H(S_D^f) \ge H(W_f) \ge l$, we obtain
	\begin{corollary}\label{cor:max} For any $E\subset D,$
		$
		H(R_E^f) \ge \frac{|E|l}{d}.
		$
	\end{corollary}
If $|E|>(d-k+1)\beta(d/l)$ then this result is better than the claim of Lemma~\ref{lemma:claim1} at the interior points.

We note that the constructions presented below do not reach the bounds proved in this section, leaving an open question of the optimal repair bandwidth for the IP repair technique.

%
%

	\section{Intermediate Processing for Evaluation Codes}\label{sec:IPC}
In this section we show that ${F}$-linear regenerating codes support repair on graphs with lower communication complexity compared to the AF strategy. In Sec.~\ref{sec:pm} we give an alternative description of node repair using the IP strategy at the MSR point for product-matrix codes and in Sec.~\ref{sec:GPM} we extend this procedure to their generalization due to Duursma and Wang \cite{Duursma2020}, which forms a new result. These two sections prepare the way for an IP node repair procedure for interior-point codes in Sec.~\ref{sec:moulin} below.
	
%
%

\vspace*{-.1in}	\subsection{Product-matrix (PM) codes} \label{sec:pm}
As our first goal, we rewrite the IP repair of PM codes originally introduced in \cite[Sec.II.A]{Patra2021} to fit the evaluation code paradigm. We begin with a brief introduction to the original description of the PM framework.
PM codes, constructed in \cite{Rashmi11}, form a family of MSR codes with parameters $[n,k,d=2(k-1),l=k-1,\beta=1,M=k(k-1)]$. 
The data file $\cF$ consists of $M$ uniformly chosen symbols from a finite field $F$. These symbols are organized to form two symmetric matrices $S_1,S_2$ of order $k-1$, each consisting of $\binom{k}{2}$ independent symbols and hence accounting for a total of
$M$ symbols. The encoding matrix $\Psi$ is taken to be an $n \times d$ matrix such that $\Psi = \begin{bmatrix}
		\Phi |\Lambda\Phi
	\end{bmatrix}$ where $\Phi$ is a $n \times (k-1)$ Vandermonde matrix with rows of the form 
$\phi_i = (1,x_i,x_i^2,\dots,x_i^{l-1}), i=1,\dots,n$ and $\Lambda = \text{Diag}(x_1^l,x_2^l,\dots,x_n^l)$ is a diagonal matrix 
where $x_1,\dots,x_n$ are distinct non-zero elements of $F$.
The encoded message is defined as $C = \Psi (S_1|S_2)^\intercal$ and the $l$ symbols of row $i$ of $C$ are stored in node $i$. Thus the
$i$th node stores the $l$-vector $\phi_iS_1+\lambda_i\phi_iS_2$.

The node repair process goes as follows: assuming that node $f \in [n]$ has failed, and the helper nodes are $D \subseteq [n]\setminus\{f\}, |D|=d$, helper node $i \in D$ sends the symbol of $F$ found as $(\phi_iS_1+\lambda_i\phi_iS_2)\phi_f^\intercal$. Since the submatrix $\Psi_D$ formed of the rows of $\Psi$ indexed by $D$ is invertible, node $f$ can calculate $S_1\phi_f^\intercal$ and $S_2\phi_f^\intercal$ from which it can compute its contents as $\phi_fS_1+\lambda_f\phi_fS_2$.

To phrase this differently, let $s_1(y,z)$ and $s_2(y,z)$ be two symmetric polynomials over ${F}$ of degree at most $k-2$ in each of the
two variables (this means, for instance, that $s_1(y,z)=s_1(z,y)$). Because of the symmetry, the total number of independent coefficients is $M$, so $s_1,s_2$ can be used to represent $\cF.$
Letting $x_1,\dots, x_n$ be distinct points of $F$, we let node $i$ store the $l$ coefficients of the polynomial 
$g^{(i)}(z)=s_1(x_i,z)+x_i^{k-1}s_2(x_i,z)$ for all $i\in[n].$
	
Using this description of the codes, the IP repair process of \cite{Patra2021} can be phrased as follows.
Let $f \in [n]$ be the failed node, let $D$ be the set of $d$ helpers, and let $A$ be a set of helper nodes of size at least $d-k+1 = k-1$. For $h \in D$ define the polynomial
	\begin{equation}\label{eq:Lagrange}
	l^{(h)}(z) = \sum_{j=0}^{d-1} l^h_j {z}^j:=\prod_{\stackrel{i \in D}{i \ne h}}\frac{z-a_i}{a_h-a_i}
	\end{equation}
of degree at most $d-1$. Then the set $A$ transmits the $l$-dimensional vector
	\begin{equation}\label{eq:eqtn4}
		{\xi}(f,A):= \sum_{h \in A}g^{(h)}(a_f)
		\left[\begin{array}{l}
			l^h_0+a_f^{k-1}l^h_{k-1}\\
			l^h_1+a_f^{k-1}l^h_{k}\\
			\hspace*{.3in}\vdots\\
			l^h_{k-2}+a_f^{k-1}l^h_{2k-3}	
		\end{array}\right].
	\end{equation}
We show that (i), the failed node can recover its value based on the vector $\xi(f,{D})$, and (ii), the intermediate nodes
can save on the repair bandwidth by processing the received information. To show (i) we prove
	\begin{lemma}\label{lemma_pm}
		The content of the failed node $f$ coincides with the vector $\xi(f,D)$, i.e.,
		$$
		g^{(f)}(z)=\sum_{i=0}^{l-1}(\xi(f,D))_i\,{z}^i.
		$$
	\end{lemma}
	\begin{proof}
Consider the polynomial $H(z) = s_1(a_f,z)+z^{k-1}s_2(a_f,z)$ and note that $\deg (H)\le 2k-3=d-1$. Thus if we write
$H(z)=\sum_{j=0}^{d-1} g_j z^j,$ then the polynomial $g^{(f)}$ defined above can be written as
\\[.05in]
\centerline{$g^{(f)}(z)=\sum_{j=0}^{k-2}(g_j+a_f^{k-1}g_{k-1+j})z^j.$}\\[.05in] 
Rephrasing, the contents of the node $f$ is
  $$
  (g_0+a_f^{k-1}g_{k-1},g_1+a_f^{k-1}g_k,\dots,g_{k-2}+a_f^{k-1}g_{2k-3})^\intercal.
  $$
At the same time, using \eqref{eq:Lagrange} we can write $H(z)$ in the Lagrange form
		$
		H(z) = \sum_{h \in D}g^{(h)}(a_f)l^{(h)}(z).
		$
The coefficient vector of this polynomial is nothing but $\xi(f,D)$.
	\end{proof}
To show part (ii) we note that the polynomials $\{l_h(z)\}_{h\in D}$ do not depend on $\cF$ and can be computed at any node in the network. So what we care to receive from the helper nodes are the multipliers $\{g^{(h)}(a_f)\}_{h\in D}$. 
Hence, for any set of helper nodes with $|A| < d-k+1$, it is gainful to send $\{g^{(h)}(a_f)\}_{h\in A}$ 
rather than the vector $\xi(f,A),$ since the former requires fewer than $l$ transmissions. At the same time,  when
$|A|\ge d-k+1,$ we can transmit the vector $\xi(f,A)$ of dimension $l$, meeting the bound of Lemma \ref{lemma:bound} and
reproducing the result from \cite{Patra2021}.

Using multilinear algebra notation (more on it in the next section), we can rephrase the code description as follows. 
The encoding is defined as a linear functional 
    $$
\phi\in(F^2\otimes S^2F^{k-1})^\ast,
    $$
where $S^2F^{k-1}$ is the second symmetric power (this is another way of saying that the encoding relies on evaluations of symmetric polynomials).
Node $i$ stores a restriction of $\phi$ to $x_i\otimes y_i\otimes F^{k-1}$, where $x_i=[1,a_i^{k-1}], 
y_i=[1, a_i, \dots , a_i^{{k-2}}].$ The contents of the failed node is a vector in the $l$-dimensional subspace 
$(x_f\otimes {y_f}\otimes F^{k-1})^\ast$, and the IP procedure recovers the coordinates of this vector in stages that correspond to moving along the repair graph toward the failed node.
A general version of this idea underlies the repair procedure in the following sections.

A general statement characterizing the savings attained by this repair procedure depends on the properties of the graph $G$ and on the choice of the helper set $D$ in relation to the failed vertex. It is possible to write it for some special graph families such as regular trees and other simple classes, as was done in \cite[Sec.~III]{Patra2021}; however we find it easier and more informative to simply illustrate the advantage of IP repair by example. 
The same approach is taken in regards to the bandwidth savings achieved by other code families considered in this paper.

	\begin{example}\label{example:pm}
{\rm	Consider the $[n=7,k=5,d=6,l=2,\beta=1,M=10]$ PM MSR code, placed on the graph shown in Fig.~\ref{fig:graph}. This graph
	should be thought of as a subgraph in a large storage network, formed by locating a helper set for the failed node.
\begin{figure}[ht]	\begin{center}\scalebox{0.6}
		{\begin{tikzpicture}[roundnode/.style={circle, draw=black, inner sep=3pt},
				rootnode/.style={circle, draw=red, very thick,  inner sep=3pt,fill=red}			]
				\node[rootnode] (1) at (0,0) {};
				\node[roundnode] (2) at (-1.5,-2) {};
				\node[roundnode] (3) at (1.5,-2) {};
				\node[roundnode] (4) at (-3,-4) {};
				\node[roundnode] (5) at (-1,-4) {};
				\node[roundnode] (6) at (1,-4) {};
				\node[roundnode] (7) at (3,-4) {};
				\path[-] (1) edge [blue,thick] (2);
				\path[-] (1) edge [blue,thick] (3);
				\path[-] (2) edge [blue,thick] (4);
				\path[-] (2) edge [blue,thick] (5);
				\path[-] (3) edge [blue,thick] (6);
				\path[-] (3) edge [blue,thick] (7);
		\end{tikzpicture}}
	\end{center}
	\caption{The graph used in our running example}\label{fig:graph}
\end{figure}
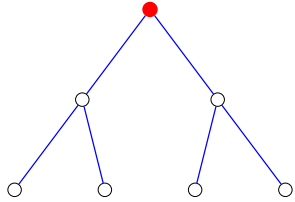		
Suppose that the root node is erased, and the remaining 6 nodes form the helper set. Since each of them contributes one symbol, the AF 
repair procedure requires transmission of $10$ field symbols over the edges to complete the repair. In particular, each of the two vertices adjacent to $v_f$ sends 3 symbols over the edge connecting it to $v_f.$ At the same time, using the IP procedure described above, these two nodes can each send only $l=2$ symbols, showing that a total of $8$ transmissions are sufficient. This shows the bandwidth saving capabilities of the IP procedure.}
\end{example}

%
%

\subsection{Linear-algebraic notation}\label{sec:tensors}
In this section we introduce elements of notation used below to define code families for which we design IP procedures of node repair.
	
For a linear space $U$ over $F$ we denote by $U^\ast$ its dual space; its elements are linear functionals of the form $\phi: U\to {F}$.
The spaces $U$ and $U^{\ast}$ have the same dimension and $(U^{\ast})^{\ast}\cong U$. A {\em restriction} of $\phi$ to a subspace
$V\subset U$ is denoted as $\phi \upharpoonright V$.
	
Let $U,V$ be linear spaces of dimensions $m$ and $n$, respectively, and let us fix bases
$\{\overline{u}_i\}_{i=1}^m$ and $\{\overline{v}_j\}_{j=1}^n.$ The tensor product of $U$ and $V$ is a linear space
	$
	U \otimes V=\{\sum_{ij}a_{ij}\overline{u}_i \otimes \overline{v}_j, a_{ij}\in {F}\}
	$
where $a_{ij} \in {F}$ and the tensors $\overline{u}_i \otimes \overline{v}_j$ form a basis in $U\otimes V$ (thus $\dim (U\otimes V)=mn$).
By definition, $u \otimes V=\{\sum_ja_j u \otimes \overline{v}_j, a_j\in {F}\}$ and $u \otimes V \subseteq U \otimes V$
The dual of a tensor product is the tensor product of duals, i.e., $(U \otimes V)^{\ast} = U^{\ast}\otimes V^{\ast}$. 
We denote by $T^pV := V^{\otimes p} $ the $p$-th tensor power of $V$. The dimension of $T^pV $ is $n^p$. 

The {\em symmetric power} $S^pV$ is a linear space of symmetric tensors, i.e., the subspace of $T^pV$ formed of the tensors invariant under transformations of the form 
$\overline{v}_1\otimes\dots\otimes\overline{v}_p\mapsto \overline{v}_{\sigma(1)}\otimes\dots\otimes\overline{v}_{\sigma(p)}$ for any
permutation $\sigma.$ We write symmetric tensors as
	$$
	\sum_{\stackrel{i_1,i_2,\dots,i_p}{1 \le i_1\le i_2\le \dots\le i_p\le n}} a_{i_1i_2\dots i_p}\overline{v}_{i_1} \odot \overline{v}_{i_2} \odot\dots \odot \overline{v}_{i_p},
	$$
where $\odot$ denotes the symmetric product and $a_{i_1i_2\dots i_p}$ are elements of $F$. 
By definition, $\dim(S^pV)=\binom{n+p-1}{p}.$ The space $S^pV$ can be  thought of as a projection 
       $$S: T^pV \rightarrow S^pV$$
that sends the tensor $\overline{v}_{i_1}\otimes \overline{v}_{i_2} \otimes \dots \otimes \overline{v}_{i_p}$ to $\overline{v}_{j_1} \odot \overline{v}_{j_2}\odot \dots \odot \overline{v}_{j_p}$ where $j_1\le j_2\le \dots \le j_p$ is a sorted copy of $i_1,i_2,\dots, i_p$. 

Finally, $x\wedge y$ denotes the exterior (alternating) product of vectors, characterized by $x\wedge y=-y\wedge x$;
hence $\overline{v}_{\sigma(1)}\wedge \overline{v}_{\sigma(2)}\wedge\dots\wedge \overline{v}_{\sigma(n)}=\sgn(\sigma)
\overline{v}_1\wedge \overline{v}_2\wedge\dots\wedge \overline{v}_n$, where $\sgn(\sigma)$ is the signature of the permutation $\sigma$. The {\em exterior power} $\Lambda^pV$ is a vector subspace of dimension $\binom np$ spanned by elements of the form $\overline{v}_{i_1} \wedge \overline{v}_{i_2} \wedge\dots \wedge \overline{v}_{i_p}, 1 \le i_1< i_2< \dots< i_p\le n$, so a vector in $\Lambda^pV$ has the form
	$$
	\sum_{\stackrel{i_1,i_2,\dots,i_q}{1 \le i_1< i_2< \dots< i_q\le n}} a_{i_1i_2\dots i_q}\overline{v}_{i_1} \wedge \overline{v}_{i_2} \wedge\dots \wedge \overline{v}_{i_q}.
	$$
The spaces $S^pV$ and $\Lambda^qV$ are formed by the action on $T^pV$ of the symmetric and alternating groups, respectively.	

By convention, $T^0V$, $S^0V$ and $\Lambda^0V$ are taken to be $F$.

%
%
	
	\subsection{Generalized PM codes}\label{sec:GPM}
An extension of the PM construction was recently proposed in \cite{Duursma2020}.
The construction of \cite[Sec.4]{Duursma2020} yields a family of MSR codes with parameters 
    $$
    n,k,d=\frac{(k-1)t}{t-1},l=\binom{k-1}{t-1}, M=t\binom kt, \quad 2\le t\le k\le n-1.
    $$
 In this section we follow the paradigm of evaluation codes to 
introduce an IP node repair procedure for this code family.
	
We start with a brief description of the code construction. Let $X={F}^t$ and $Y={F}^{k-t+1}$. Let $L:=X \otimes S^tY$ and note that
$\dim(L)=M.$ The encoding $\phi:L\to F^{nl}$ is an $F$-linear map. To define a concrete encoding procedure, we fix a basis in $L^*$ and 
let the coordinates of $\phi$ be the contents of the stored data. 

To support the data reconstruction and node repair tasks, we further choose, for each $i \in [n],$ a pair of vectors $x_i\in X$ and $y_i\in Y$ 
such that
	\begin{itemize}
		\item[(i)] Any $t$-subset of $x_i$'s spans $X$.
		\item[(ii)] Any $(k-t+1)$-subset of $y_i$'s spans $Y$.
		\item[(iii)] Any $d$ subspaces $x_i\otimes y_i \odot S^{t-2}Y$ span $X \otimes S^{t-1}Y.$
	\end{itemize}
The first two properties enable data reconstruction, while the node repair property depends on the third condition \cite{Duursma2020}. 

 With these assumptions, the contents of node $i$ correspond to the restriction 
	$\phi \!\upharpoonright\! x_i \otimes y_i \odot S^{t-1}Y \in  (x_i \otimes y_i \odot S^{t-1}Y)^\ast.$ This is consistent with the
code parameters: indeed, an element in $(x_i \otimes y_i \odot S^{t-1}Y)^\ast$ is completely described by its evaluations on a basis of the space $x_i \otimes y_i \odot S^{t-1}Y,$ which requires storing exactly $l=\binom{k-1}{t-1}$ evaluations.

As before, let $f \in [n]$ be the (index of the) failed node and let $D \subseteq [n]\setminus\{f\}$ be the helper set. 
Note that we wish to recover the restriction	$\phi \upharpoonright x_f \otimes y_f \odot S^{t-1}Y.$ 
Choose a basis for $x_f \otimes y_f \odot S^{t-1}Y$ and let 
  $
x_f \otimes y_f \odot (\overline{y}_{i_1}\odot \dots \odot \overline{y}
_{i_{t-1}})
  $ 
  be one of the basis vectors. Let 
    $$
    \{\underline{y}_{j_1} \odot \dots \odot \underline{y}_{j_{t-2}}, 1\le j_1\le j_2\dots\le
j_{t-2}\le n\}
    $$
 be a basis of $S^{t-2}Y$. 
The helper node $i \in D$ transmits to the failed node the restriction of $\phi$ to the set of vectors 
$\{x_i \otimes y_i \odot (\underline{y}_{j_1} \odot \dots \odot \underline{y}_{j_{t-2}})\odot y_f\}$.

\vspace*{.1in} It becomes easier to think of the above construction once we connect it with PM codes described in Sec.~\ref{sec:pm}.
For that, take $t=2.$ In this case, the file size is 
$$
\dim(L) = \dim(X \otimes S^2Y) = \dim(F^2 \otimes S^2F^{k-1}) = k(k-1).
$$
Node $i$ stores $\phi \upharpoonright (x_i \otimes y_i \odot Y)$, i.e, $\phi$ evaluated at a basis of $x_i \otimes y_i \odot Y,$ which requires 
storing exactly $\dim(Y)=k-1$ symbols. Each node can calculate the symbol $\phi(x_i\otimes y_i \odot y_f)\in F$. Now notice that $d$ vectors $\{x_i \otimes y_i\}$ span $X \otimes Y,$ and so $d$ values $\phi(x_i\otimes y_i \odot y_f)$ account for the evaluations of $\phi$ on $X \otimes Y \odot y_f$. From this set of evaluations we can calculate $\phi$ on $x_f \otimes Y \odot y_f$ which by the symmetric product property is the 
same as $x_f \otimes y_f \odot Y$. These evaluations form the contents of the failed node.
   
The IP repair for this construction works as follows. By (iii) above we can write
	\begin{align*}
		&x_f \otimes y_f \odot (\overline{y}_{i_1}\odot \dots \odot \overline{y}_{i_{t-1}})
		= x_f \otimes (\overline{y}_{i_1}\odot \dots \odot \overline{y}_{i_{t-1}}) \odot y_f\\
		&= \sum_{i \in D}\sum_{j_1,\dots,j_{t-2}}a_{i,j_1,\dots, j_{t-2}}x_i \otimes y_i \odot \cY_{j_1,\dots,j_{t-2}}\odot y_f,
	\end{align*}
where we denoted $\cY_{j_1,\dots,j_{t-2}}=\underline{y}_{j_1} \odot \dots \odot \underline{y}_{j_{t-2}}.$ Again similarly to the PM codes, any set $A \subseteq D$ with $|A| \ge d-k+1$ can transmit the following single evaluation of $\phi$
along the path to $f$:
	\begin{align*}
		&\phi\Big(\sum_{i \in A}\sum_{j_1,\dots,j_{t-2}}a_{i,j_1,\dots, j_{t-2}}x_i \otimes y_i \odot \cY_{j_1,\dots,j_{t-2}}\odot y_f\Big)\\
		&=\sum_{i \in A}\sum_{j_1,\dots,j_{t-2}}a_{i,j_1,\dots, j_{t-2}}\phi(x_i \otimes y_i \odot \cY_{j_1,\dots,j_{t-2}}\odot y_f).
	\end{align*}
	This can be done for all basis vectors of the chosen basis of $x_f \otimes y_f \odot S^{t-1}Y,$ and that requires $l=\binom{k-1}{t-1}$ transmissions, which matches the lower bound of Lemma~\ref{lemma:bound}. Note that the AF repair would require any set $A$
of helpers to transmit $\beta |A|$ symbols of $F$, which is greater than $l$ for $|A|> d-k+1.$ 

We have shown that IP repair can outperform direct relaying.  Let us give an example to support this claim (note also the remark
before Example~\ref{example:pm}).
\begin{example}\label{example:gen_pm}
{\rm	Consider the use of generalized PM codes for the graph shown in Fig.~\ref{fig:graph}. Suppose that $t=3$, i.e., 
the code parameters are $[n=7,k=5,d=6,l=6,\beta=3,M=30]$.
	Again considering the repair of the root node, the AF repair procedure would require transmission of 
$3\cdot(1+1+1+1+3+3)=30$ symbols while the IP procedure requires only $3\cdot(1+1+1+1+2+2)=24$ symbol transmissions. It is
easy to construct many other similar examples.}
\end{example}
%
%
\subsection{Operations on product spaces}\label{sec:tensors2}
In preparation for discussing IP repair with Moulin codes in the next section, we define (following \cite{Duursma2021})
two operations on tensor product spaces.
Let $V={F}^{d-k}$, $W={F}^k,$ and $U=V\oplus W\cong{F}^d$. We shall be dealing with spaces of the form $T^pV \otimes V \otimes \Lambda^qW$ and $T^pV \otimes W \otimes \Lambda^qW$ where $p+q=s-1$ and $p,q\ge 0.$ Note that 
   $$
   T^pV \otimes V \otimes \Lambda^qW \oplus T^pV \otimes W \otimes \Lambda^qW = T^pV \otimes U \otimes \Lambda^qW,
   $$
and hence there are natural inclusion maps from each of these spaces to their direct sum, as well as natural projection maps from the direct sum to these spaces.

Define the \emph{co-wedge product} operator inductively as follows: 
	\begin{align*}
		\nabla : T^pV \otimes \Lambda^1W &\rightarrow T^pV \otimes W\\
		\nu \otimes w_1 &\rightarrow \nu \otimes w_1\\
		\nabla: T^pV \otimes \Lambda^2W &\rightarrow T^pV \otimes W \otimes \Lambda^1W\\
		\nu \otimes w_1 \wedge w_2 &\rightarrow \nu \otimes w_1\otimes w_2 - \nu \otimes w_2 \otimes w_1\\
		\nabla : T^pV \otimes \Lambda^{q+1}W &\rightarrow T^pV \otimes W \otimes \Lambda^qW\\
		\nu \otimes \omega \wedge w_1 &\rightarrow \nabla(\nu \otimes \omega) \wedge w_1+(-1)^q\nu \otimes w_1 \otimes \omega,
	\end{align*}
where on the last line $\omega\in \Lambda^{q}W$ and $w_1\in W.$
Thus, as a result of applying $\nabla,$ the degree of the wedge product decreases by one. For tensors of higher ranks, $\nabla$ applies term-wise, and the images are added. The operator $\nabla$ is clearly linear. Next we define the {\em coboundary operators (differentials)} which increase the degree of tensors. For any $v \in V$ define the linear transformation inductively:
	\begin{gather*}
		\partial^V_v : \Lambda^qW \rightarrow U \otimes \Lambda^qW\\
		\omega \rightarrow 0\\
		\partial^V_v : U \otimes\Lambda^qW \rightarrow T^1V \otimes U \otimes \Lambda^qW\\
		\hspace*{.2in}u \otimes \omega \rightarrow v \otimes u \otimes \omega\\
		\partial^V_v : T^pV \otimes U \otimes\Lambda^qW \rightarrow T^{p+1}V \otimes U \otimes \Lambda^qW\\
		\hspace*{.4in}\nu \otimes u \otimes \omega  \rightarrow \partial^V_v(\nu) \otimes u \otimes \omega+ (-1)^p\nu\otimes v\otimes u \otimes \omega
	\end{gather*}
for all $p\ge 1,q\ge 0$. Note that when $q=0$ we take $\Lambda^0=F.$
In the other direction, for every $w \in W$ and $p\ge 0, q\ge 1$ define the mappings
	\begin{align*}
		\partial^W_w: T^pV \otimes U &\rightarrow T^pV \otimes U \otimes \Lambda^1W\\
		\nu \otimes u &\rightarrow(-1)^p\nu\otimes u\otimes w\\
		\partial^W_w: T^pV \otimes U \otimes \Lambda^qW&\rightarrow T^pV \otimes U \otimes \Lambda^{q+1}W\\
		\nu \otimes u \otimes \omega &\rightarrow(-1)^{p+q}\nu\otimes u\otimes \omega \wedge w..
	\end{align*}
Finally for $u\in U$ such that $u=v+w, v\in V,w\in W$, define 
	$$
	\partial^U_u = \partial^V_v+\partial^W_w.
	$$
Thus, the overall diagram has the form
\begin{center}
\begin{tikzcd}
{\dots}\rar[shorten <= 4em, very near end, "\partial_w^W"]&T^{p+1}V\otimes U\otimes \Lambda^q W\arrow[r, shorten >=0.3in, near start, "\partial_w^W"]&{\dots}\\
T^PV\otimes U\otimes \Lambda^{q-1}W \arrow[r, "\partial_w^W"] \arrow[u, "\partial_v^V"]
&T^pV\otimes U\otimes \Lambda^q W \arrow[u, "\partial_v^V"]\arrow[r, "\partial_w^W"]
&T^pV\otimes U\otimes \Lambda^{q+1}W\arrow[u, "\partial_v^V"]\\
{\dots}\arrow[u, "\partial_v^V"]\rar[shorten <= 4em, very near end, "\partial_w^W"] &T^{p-1}V\otimes U\otimes \Lambda^q W\arrow[u, "\partial_v^V"]\arrow[r, shorten >=0.4in, near start, "\partial_w^W"] &{\dots}\arrow[u, "\partial_v^V"]
\end{tikzcd}.
\end{center}
Except for $U$, this diagram follows the standard construction of the tensor product of chain complexes \cite[Sec.10.1]{Rotman2009}, and the differentials
satisfy the usual relations: $(\partial_v^V)^2=0, (\partial_w^W)^2=0$ for all $v\in V, w\in W$, and $\partial_v^V\partial_w^W+\partial_w^W\partial_v^V=0.$

	\subsection{IP for Interior Point Codes}\label{sec:moulin}
In this section we switch attention from MSR codes to a class of intermediate-point evaluation codes introduced recently by Duursma et al. in \cite{Duursma2021} (see also \cite[Sec.~7.2]{Ramkumar2022}). Let $s$ be an integer such that $n-1 \ge d \ge k \ge s-1\ge 1$. The family of {\em Moulin codes} 
that we discuss has parameters $[n,k,d,l, \beta, M]$ that satisfy the relations
	\begin{equation}
	\left.\begin{array}{@{\hspace{-.2in}}c}	\textstyle{l = \sum_{p+q=s-1}(d-k)^p\binom{k}{q}} \\[.1in]
		\textstyle{\beta = \sum_{p+q=s-2}(d-k)^p\binom{k-1}{q}} \\[.1in]
		\textstyle{M = \sum_{p+q=s-1}d(d-k)^p\binom{k}{q} - \sum_{p+q=s}(d-k)^p\binom{k}{q}}, 
		\end{array}\right\}\label{eq:M}
	\end{equation}
where $p\ge 0,q\ge0$ throughout. 	 
	 While the general idea of implementing IP for this code family is the same as before (node contents are given by restrictions of linear maps to subspaces), the detailed description relies on the operations on tensor products introduced above. 

For a fixed $s$ satisfying the constraints above, the file $\cF$ is chosen to be an element $\phi$ of the dual space 
	\begin{equation}\label{eq:ds}
	\bigoplus_{p+q=s-1}(T^pV \otimes U \otimes \Lambda^qW)^\ast,
	\end{equation}
where $V,W,U$ are as in the previous section.	
The parity checks of the code correspond to the
condition of having the following diagrams commute:
\begin{center}
\begin{tikzcd}[column sep=-.3in]
T^pV \otimes \Lambda^{q+1}W \arrow[rr, "\phi"] \arrow[dr, "\nabla"] & & F\\
&T^pV \otimes W \otimes \Lambda^q W \arrow[ur, "\phi"] & \\
\end{tikzcd}
\end{center}
\vspace*{-.2in} for all $p\ge 1,q\ge 0$ with $p+q=s-1$, and
\begin{equation}\label{eq:root-check}
	\begin{tikzcd}[column sep=-.1in]
		\Lambda^{q+1}W \arrow[rr, " 0 "] \arrow[dr, "\nabla"] & & F\\
		& W \otimes \Lambda^q W \arrow[ur, "\phi"] & \\
	\end{tikzcd}\hspace*{.3in}
	\begin{tikzcd}[column sep=.2in]
		T^pV  \arrow[rr, "\phi"] \arrow[dr, "\nabla"] & & F\\
		&0 \arrow[ur, "0"] & \\
	\end{tikzcd}
\end{equation}

\vspace*{-.2in}\noindent for $p=0$ and $q=-1$, respectively. 

The file size equals the dimension of the direct sum of the vector spaces \eqref{eq:ds} minus the dimension of the parity check space, 
which is exactly $M$ in \eqref{eq:M}.
To each node $i\in[n]$ we associate a vector $u_i \in U$ such that any $d$ of these vectors span $U$ and any $k$ vectors 
span $U/V$ {under the quotient map $U\to U/V$}.
The $i$-th node stores the following restriction of the mapping $\phi$:
	$$
	\phi \upharpoonright\!\! \bigoplus_{p+q=s-1}(T^pV \otimes u_i \otimes \Lambda^qW).
	$$
	The size $l$ of the node equals $\dim(T^pV \otimes u_i \otimes \Lambda^qW),$ given by $l$ in \eqref{eq:M}. 
	
Now suppose that node $f \in [n]$ fails and we are provided with a set $D \subseteq [n]\setminus\{f\}$ of $d$ helpers. 
Each node $h \in D$ {provides the restrictions of its contents to coboundaries:}
  \begin{equation}\label{eq:restriction}
	\phi \upharpoonright \partial_{u_f}^U(T^pV \otimes u_h \otimes \Lambda^qW)
  \end{equation}
for each pair $p,q$ with $p+q = s-2$. 
We shall need the following result.
	\begin{lemma}[\cite{Duursma2021}, Thm.~4.1]\label{lemma:dd}
		For all possible $p,q \ge 0$, such that $p+q =s-2$ and all $\nu \in T^pV, \omega \in \Lambda^qW$, we have
		$$
		\phi(\partial_{u_f}^U(\nabla(\nu \otimes \omega)))- \phi(\partial_{u_f}^U(\nu \otimes \omega)) = (-1)^p\phi(\nu \otimes u_f \otimes \omega).
		$$
If $p=0$, then $\phi(\partial_{u_f}^U(\omega))=0$ due to \eqref{eq:root-check}.
	\end{lemma}
The right-hand side of the above equation is one coordinate of the failed node, and the left-hand side can be computed from
\eqref{eq:restriction}.

 The statement of the next lemma appears in \cite{Duursma2021} without a proof (as a statement in the proof of \cite[Thm.~4.1]{Duursma2021}). We include the proof here to set up the notation.
	\begin{lemma}\label{lemma:lin}
		1) For all possible $p\ge 1,q \ge 0$ such that $p+q =s-1$ and all $\nu \in T^pV, \omega \in \Lambda^qW$, the tensor $\nu \otimes \omega$ is contained in the linear span of the union of the spaces $\{T^{p'}V \otimes u_h \otimes \Lambda^{q'}W\}_{h\in D, p'+q'=s-2}$.
		
		2) For all possible $p,q \ge 0$, such that $p+q =s-1$, for all $\nu \in T^pV, \omega \in \Lambda^qW$, $\nabla(\nu \otimes \omega)$ is contained in the linear span of the union of the spaces $\{T^{p'}V \otimes u_h \otimes \Lambda^{q'}W\}_{h\in D, p'+q'=s-2}$.
	\end{lemma}
\begin{proof}
		1) Fix $p_1\ge 1,q_1>0$ such that $p_1+q_1=s-1$. Let $\nu \in T^{p_1}V$ and $\omega \in \Lambda^{q_1}W$. Fix a basis $\{\overline{\nu}_i\otimes u_h \otimes \overline{\omega}_i\}_{i=1}^{(d-k)^{p_1-1}\binom{k}{q_1}}$ of $T^{p_1-1}V \otimes u_h \otimes \Lambda^qW$. Since 
		the set $\{u_h\}_{h\in D}$ spans $U$, we can write 
		\begin{align*}
			\nu \otimes \omega &= \overbrace{(\nu_1 \otimes \dots \otimes \nu_{p_1-1})}^{T^{p_1-1}V}\otimes \underbrace{\nu_{p_1}}_{U} \otimes \overbrace{\omega}^{\Lambda^qW},
			\end{align*}
			and hence $\nu\otimes\omega$ is an element of $T^{p_1-1}V\otimes U \otimes \Lambda^{q_1}W$. So we can write
			$
			\nu \otimes \omega= \sum_{i,h}a_{i,h}(\overline{\nu}_i\otimes u_h \otimes \overline{\omega}_i).
		$
		
		2) Similarly, for a basis $\{\underline{\nu}_j\otimes u_h \otimes \underline{\omega}_j\}_{j=1}^{(d-k)^{p_1}\binom{k}{q_1-1}}$ of $T^{p_1}V \otimes u_h \otimes \Lambda^{q_1-1}W$, we can write
		$
			\nabla{(}\nu \otimes \omega{)} = \sum_{j,h}b_{j,h}(\underline{\nu}_j\otimes u_h \otimes \underline{\omega}_j).
		$
\end{proof}

Our main statement in this part is the next lemma, which justifies the IP repair procedure.
	\begin{lemma} Let $A \subseteq D.$ For the repair of $f$, it is sufficient for the nodes in the set $A$ to transmit $l$ symbols.
	\end{lemma}
	
	\begin{proof}
		Fix $p,q \ge 0$ such that $p+q=s-1$. Let $\nu \in T^pV, \omega \in \Lambda^qW$. If $p\ge 1$ then by parts (1) and (2) of Lemma \ref{lemma:lin}, we have
		\begin{align*}
			\phi(\partial_{u_f}^U(\nabla(\nu \otimes \omega)))&- \phi(\partial_{u_f}^U(\nu \otimes \omega))\\
			&=
		\sum_{j,h}b_{j,h}\phi(\partial_{u_f}^U(\underline{\nu}_j\otimes u_h \otimes \underline{\omega}_j)) 		
			- \sum_{i,h}a_{i,h}\phi(\partial_{u_f}^U((\overline{\nu}_i\otimes u_h \otimes \overline{\omega}_i)).
		\end{align*}
		Note that if $p=0$ then the second term on the LHS $\phi(\partial_{u_f}^U(\nu \otimes \omega))$ is already 0 by \eqref{eq:root-check} and we simply write using Lemma~\ref{lemma:lin}(2)
		$$ 
			\phi(\partial_{u_f}^U(\nabla(\nu \otimes \omega)))=
			\sum_{j,h}b_{j,h}\phi(\partial_{u_f}^U(\underline{\nu}_j\otimes u_h \otimes \underline{\omega}_j)) 	.
		$$ 
	By Lemma~\ref{lemma:dd}, the LHS equals $(-1)^p\phi(\nu \otimes u_f \otimes \omega),$ and we have recovered one symbol of the failed node. For this, the set $A$ need to transmit the element 
		$$
	\sum_{h\in A}\!\Big[\!\sum_jb_{j,h}\phi(\partial_{u_f}^U(\underline{\nu}_j\otimes u_h \otimes \underline{\omega}_j)) - \sum_{i}a_{i,h}\phi(\partial_{u_f}^U(\overline{\nu}_i\otimes u_h \otimes \overline{\omega}_i)\!\Big].
		$$
		Doing this for any fixed basis  $\{\nu,\omega\}$ of $T^pV\otimes u_f\otimes \Lambda^qW,$ for all values of $p,q$, requires the set 
		$A$ to transmit a total of $l$ symbols. 
	\end{proof}
Observe that whenever $|A|\ge \lceil\frac l\beta\rceil$, the IP protocol given by this lemma results in communication savings
compared to the AF repair.

\begin{example}\label{example:moulin}
	{\rm We again use the graph in Fig.~\ref{fig:graph} to demonstrate the savings in required transmission bandwidth for repair. The parameters of the code construction are: $[n=7,k=5,d=6, s=4, l=26,\beta=11,M=125]$. Note that this code satisfies $d\beta = 66 > l > (d-k+1)\beta = 22$ and operates at an interior point of the storage-bandwidth trade-off curve.

	Considering the repair of the root node, the AF repair procedure would require a total transmission of $11\times 4+33\times 2=110$ symbols while performing the IP procedure in the nodes neighboring the failed node results in a total of $11\times 4+26\times 2 = 96$ symbol, saving $7$ transmissions at each of the two neighbors. Note that the lower bound from Lemma~\ref{lemma:bound} says that the minimum transmission bandwidth is at least $11 \times 4+ 22 \times 2 = 88$ symbols (in this case Cor.~\ref{cor:max} gives a weaker result).}
\end{example}

%
%
	
	\section{IP repair for other code families}
\subsection{Determinant codes}\label{sec:det} Determinant codes \cite{Elyasi2016,Elyasi2019} represent another well-known family of intermediate-point regenerating codes. Of several versions of the construction presented by the authors, we follow the one 
appearing in \cite{Elyasi2019}. To remind ourselves of the general context, let $n,k$ be fixed, and let $d=k.$ Recall that the tradeoff curve \eqref{eq:sbt} isolates
a polygon on the bandwidth-storage plane called the {\em exact repair region}. In particular, as shown in \cite{Elyasi2016}, for $d=k$ the exact repair region is a convex hull of $k$ points given by $l_m=\binom{k}{m}, \beta_m=\binom{k-1}{m-1}, M_m=m\binom{k+1}{m+1}$, for $m=1,2,\dots,k$ and these points are achieved by the determinant code construction. 
Moreover, the intermediate points of the bound \eqref{eq:sbt} can be achieved by space sharing. In this section we observe 
that determinant codes can be easily adapted to support the IP technique.

Let us begin with a brief description of the code construction (see the original paper \cite{Elyasi2019} for more details), noting that linearity 
of the codes is again at the root of this application. Fix some $m\in[k].$ The symbols of the data file $\cF$ are arranged in two matrices, denoted below by $V$ and $W$, of dimensions $l_m\times d$ and $l_{m+1}\times d$, respectively. The rows are of $V$ are indexed by the $m$-subsets of the set $[d]:=\{1,2,\dots,d\},$ the rows of $W$ are indexed by the $(m+1)$-subsets, and the columns of either matrix 
are indexed 
by the elements of $[d].$ Accordingly we label the data symbols with two subscripts $j$ and $A$, where $j\in[d]$ and $A\subset [d].$ 
Write these symbols as
   \begin{gather*}
   \cV = \{v_{A,j} \in F\mid A \subset [d], |A| = m, j \in A\}\\
   \cW = \{w_{S,j} \in F\mid S \subset [d], |S| = m+1, j \in S, \tau_S(j)\le m\},
   \end{gather*}
where $\tau_S(j) = |\{i \in S: i \le j\}|$, i.e., in $\cW$ we do not assign a data element to the largest index within each of the subsets $S$. Instead, the largest location within each $S$ is assigned the value that fulfills the parity check equation
  $$
   \sum_{j\in S}(-1)^{\tau_S(j)}w_{S,j}=0,
   $$
yielding a total of $\binom d{m+1}$ parity symbols.   
Now assign the data symbols (and in the case of $W$ also the parity symbols) to the corresponding places in the matrices $V$
and $W$, writing them in the locations indexed by the elements of the subsets, and fill the remaining empty places in the matrices
with zeros. 

In the next step the matrices $V$ and $W$ are used to construct an $l_m \times d$ data matrix $D$, whose rows are again indexed 
by the sets $A\subset [d]$ and columns by $[d]$, as follows
$$d_{A,j} = \begin{cases}
	v_{A,j} & \mbox{ if }j \in A\\
	w_{A\cup\{j\},j} & \mbox{ if }j \notin A
\end{cases}.
$$ 
Note that $|\cV|+|\cW|=M_m$ and $M_m+\binom d{m+1}=l_m d,$ the number of matrix elements in $D$.
Finally to obtain a codeword that corresponds to the data file, we multiply $D$ by a $d \times n$ matrix $\Phi$ such that each $k$-subset of
its rows has full rank over $F$, for instance a Vandermonde matrix. This yields an $l_m\times n$ codeword matrix $C$ over $F.$

Next we describe the node repair procedure suggested in \cite{Elyasi2019}. For a matrix $G$  denote its $i$th row by $G_{i,:}$ and $i$th column by $G_{:,i}$. Thus, the contents of the $i$th node (the $i$th coordinate of the codeword $C$) is given by
$C_{:,i}=D\Phi_{:,i}.$ Without loss of generality assume that node 1 has failed and nodes in the set $H=\{2,3,\dots, d+1\}$ are used as the helper nodes. Define the $l_{m-1}\times l_m$ matrix $R$, whose rows and columns are indexed by $(m-1)$- and $m$-subsets of $[d]$, as follows:
      $$ 
   R_{B,A} = \begin{cases}
	(-1)^{\tau_A(j)}\phi_{j,1} & \mbox{ if } \exists y \mbox{ s.t. }A=B \cup \{j\}\\
	0 & \mbox{otherwise},
      \end{cases}
      $$
where $\phi_{j,1}$ is an element of the matrix $\Phi.$      
 We note that the matrix $R$ depends only on the index of the failed node and can be pre-computed independently at each helper node.
To perform repair, the failed node downloads from helper node $i\in H$ the vector $R C_{:,i}$. The dimension of this vector is $l_{m-1}$, so on the face of it, the required size of the download exceeds the allotted repair bandwidth $\beta_m$. However, by \cite[Prop.~1]{Elyasi2020} the rank of the matrix $R$ is at most $\beta_m,$ so as many symbols suffice to communicate the vector $R C_{:,i}$ from the $i$th (helper) node to the failed node.

At the failed node, the vectors $R C_{:,i},i\in H$ are written as columns of a $l_{m-1} \times d$ matrix $T$. Since $C_{:,i}=D\Phi_{:,i},$ we can write $T=RD\Phi_{:,H}$, where $\Phi_{:,H}$ is the submatrix of $\Phi$ formed of the columns $2,3,\dots,d+1.$ By construction, $\Phi_{:,H}$ is
invertible, and the failed node can find the $l_{m-1}\times d$ matrix $R\, D$.
These elements suffice to recover the contents of the failed node as shown in the following lemma due to \cite{Elyasi2020}, Prop.~2. Since our modification of the repair procedure depends on this statement, we include a proof in the appendix.
\begin{lemma}\label{lemma:cA1}
For any $A \subset [d], |A| = m$, 
\begin{equation}\label{eq:det_rep}
C_{A,1} = \sum_{i \in A}(-1)^{\tau_A(i)}R_{A\setminus\{i\},:}D_{:,i},
\end{equation}
where $R_{A\setminus\{i\},:}$ is the row of $R$ with index $A\backslash\{i\}.$
Thus the contents of the failed node can be recovered from the matrix $RD.$
\end{lemma}

Note that $R_{A\setminus\{i\},:}D_{:,i}$ is an element in the product $R\cdot D,$ which is exactly the information available to the failed node.
The point that we wish to make is that the described repair procedure can be modified to support IP repair for determinant codes used on a graph. 
To formulate it, we need some notation. Let $\overline{C}_H = \begin{bmatrix} C_{:,2}^\intercal  C_{:,3}^\intercal  \ldots  C_{:,d+1}^\intercal \end{bmatrix}^\intercal$ be a $dl_m$-dimensional column vector obtained by 
concatenating columns $2,3,\dots,d+1$ of $C$. Define $l_m\times l_m$ matrices
$W^{(i)}, i=1,\dots,d,$ whose rows are indexed by $m$-subsets of $[d].$ For a given $m$-subset $A\subset[d]$
the $A$-th row of $W^{(i)}$ is defined as:
   $$ 
   W^{(i)}_{A,:} = \begin{cases}
     (-1)^{\tau_A(i)}R_{A\setminus\{i\},:} & \mbox{ if }i \in A\\
      {\bf 0} & \mbox{ otherwise }
         \end{cases}.
    $$

\begin{proposition}
The contents of the failed node can be found as
  \begin{equation}\label{eq:pip}
  C_{:,1} = U\overline{C}_H = \begin{bmatrix} U^{(1)}  U^{(2)}  \ldots  U^{(d)} \end{bmatrix}\overline{C}_H,
  \end{equation}
where $U$ is an $l_m \times dl_m$ matrix determined by the contents of the helper set $\Phi_{:,H}$.
\end{proposition}
\begin{proof}

As before, let $C_{:,H}$ and $\Phi_{:,H}$ be the submatrices of the matrices $C$ and $\Phi$ with columns indexed by the set $H$, so
   $
   C_{:,H} = D\,\Phi_{:,H},
   $
or 
   $$
   D=C_{:,H}\,\Phi_{:,H}^{-1}.
   $$    
Similarly to $\overline{C}_H,$ let $\overline{D} = \begin{bmatrix}	D_{:,1}^\intercal  D_{:,2}^\intercal  \ldots  D_{:,d}^\intercal \end{bmatrix}^\intercal$ 
be the flattened matrix $D,$ written as a column vector of length $dl_m$.
 Let $\overline{\Phi}_H=(\Phi_{:,H}^{-1}\otimes I_{l_m})^\intercal$ be the $dl_m \times dl_m$ block matrix. 
Then
     \begin{equation*}
     \overline{D} = \overline{\Phi}_H\overline{C}_H.
     \end{equation*}
 Now, according to this relation and \eqref{eq:det_rep}, 
    $$
 C_{:,1} = W\overline{D} = U\overline{C}_H,
    $$
    where $W=\begin{bmatrix} W^{(1)}  W^{(2)}  \dots  W^{(d)} \end{bmatrix}$ and $U=\begin{bmatrix} W^{(1)}  W^{(2)}  \dots  W^{(d)} \end{bmatrix}\overline{\Phi}_H,$ proving \eqref{eq:pip}. Moreover, the matrix $U$ depends only on $\Phi_{:,H}$, and the proof is complete.
\end{proof}
As before, representation \eqref{eq:pip} supports ``pipeline'' repair of the contents of $C_{:,1}$, which can be spread across the nodes
of the helper set. Specifically,  instead of transmitting $|E|\beta_m$ symbols, any set $E$ of helper nodes can only transmit the vector
   $$
   \sum_{i\in E}U^{(i)}C_{:,i},
   $$ 
which requires sending a total of $l_m$ symbols over the edges leaving $E$ along the shortest path toward the failed node. 
Hence whenever $|E| > \frac{l_m}{\beta_m}$, this procedure accounts for savings in the repair bandwidth over simple forwarding (the AF repair).

\begin{example}\label{example:determinent} 
	{\rm Going back to our running example in Fig.~\ref{fig:graph}, choose $m=3,$ then the code parameters are $n=7,k=d=6,M=105$, and
we obtain an interior-point code operating at the point $(l_m=20,\beta_m=10)$ of the trade-off curve.
As before, suppose our goal is to repair the root node, while all the remaining 6 nodes serve as helpers.
The AF repair procedure would require transmission of $4\cdot10+2\cdot30=100$ symbols while performing 
intermediate processing at the nodes adjacent to the root node results in a total of $4\cdot10+2\cdot20=80$ symbol transmissions, saving $10$ 
transmissions at each of the two nodes. In this case the bounds in Lemma~\ref{lemma:bound} and Cor.~\ref{cor:max} both suggest that each of the two nodes closest to the root should send at least 10 symbols, resulting in a gap of 20 to the IP construction.}
\end{example}

\subsection{Cascade codes}\label{sec:cascade}
 A family of regenerating codes based on determinant codes was introduced in \cite{Elyasi2020}. For an
integer parameter $\mu, 1\le\mu\le k$ the parameters
of {\em cascade codes} are
\begin{equation*}
	\begin{array}{@{\hspace{-.2in}}c}	{l = \sum_{m=0}^\mu (d-k)^{\mu-m}\binom{k}{\mu}} \\[.1in]
		{\beta = \sum_{m=0}^\mu (d-k)^{\mu-m}\binom{k-1}{m-1}} \\[.1in]
		{M = \sum_{m=0}^\mu k(d-k)^{\mu-m} \binom{k}{m} - \binom k{\mu+1}. }
		\end{array}
	\end{equation*}
We note that the parameters of this code family coincide with the parameters of Moulin codes \eqref{eq:M}, as can be seen by setting
$s-1=\mu$ in \eqref{eq:M} (the families themselves are different; see \cite{Duursma2021} for more on this). Also, setting $\mu=1$ or $\mu=k$
yields the MSR and MBR points of the tradeoff curve \eqref{eq:sbt}, respectively, while otherwise the codes operate at interior points. Finally,
setting $d=k$ recovers the parameters of determinant codes. 

A cascade code is formed by stacking together several determinant codes with different values of the parameter $\mu$, called the mode 
of the component codes. The encoder mapping is again linear and therefore can be accomplished by multiplying a data matrix by the encoder
matrix. Cascading together determinant codes of varying modes enables the authors of \cite{Elyasi2020} to obtain codes for all values
$d\ge k$ as opposed to $d=k$ in the previous section. To provide the functionality of data recovery from any $k$ nodes, the data symbols 
are encoded into several copies of the constituent determinant codes via the process called symbol injection. The details of the construction
are too involved to be presented here, and we refer the readers to the original paper. An important point for us is that repair of the
failed node is performed by concatenating the repair data obtained independently from the constituent determinant codes. Since each of them
supports the IP processing, the overall construction can be also placed on the vertices of the graph to be repaired by combining
the fragments computed by the helpers along the path from them to the failed node in the repair tree.

\section{Data retrieval for codes on graphs}\label{sec:retrieval}
In Part I of this paper, \cite{Patra2021}, we have only focused on the node repair aspect of regenerating codes on graphs, circumventing the more frequently occurring task of data retrieval. The reason behind this is that at the MSR point, which was the main focus of \cite{Patra2021}, the task of data retrieval from a regenerating code defined on an arbitrary graph becomes trivial. Since MSR codes are Maximum Distance Separable (MDS) by definition, any set $A$ of $k$ or fewer nodes has to transmit $|A|\cdot l$ symbols and there is no hope of compressing this any further. This implies that in the restricted connectivity setting, when the Data Collector (DC) does not have direct access to $k$ nodes, standard relaying of data is optimal. The situation changes when we lift the MSR constraint.
 
\begin{example}
{\rm The following example shows that for MBR code families the task of data recovery can be accomplished by downloading 
fewer than $kl$ symbols from the chosen subset of $k$ nodes. Consider the family of {\em polygonal codes} \cite{Shah2012}, which
closely follows the definition of MBR codes. The parameters of the family are $n,k,d=n-1, l=n-1, \beta=1, M=k(n-1)-k(k-1)/2.$
To construct the code, fix $n$ and $k$ and choose an MDS code of length $N=n(n-1)/2$ and dimension $M$ over a field $F_q$ of size $q\ge N.$ The encoding mapping of the polygonal code $\cC_n$ is formed of two steps. In the first step we encode $M$ symbols of the file $\cF$ into a codeword of the MDS code. The length $N$ is chosen
to support a bijection between the coordinates of the codeword and edges of a complete graph $K_n,$ so we place each encoded symbol
on an edge of the graph. Each vertex of $K_n$ models a storage node. To complete the data placement in the system, we assign to each 
node the symbols written on the edges incident to it. Thereby, every node carries $n-1$ symbols of the encoding, which 
matches the parameters of the code $\cC_n.$ 

To reconstruct the file $\cF$, the DC accesses an arbitrary subset $K$ of $k$ nodes of the graph, which in total contain $k(n-1)$ symbols of the codeword. Since each pair of nodes shares one common symbol, the DC downloads $l,l-1,\dots, 1$ symbols from 
the nodes in $K$ (taken in some fixed order). This yields a total of $M$ stored symbols, so the DC is able to 
recover the MDS codeword and therefore also the file. Note a saving of $\binom k2$ symbols compared to downloading the entire contents of the $k$ nodes.}
\end{example}

In this section we elaborate on this example in two ways. First, in Lemma~\ref{lemma:bound2} below  we derive a lower bound on the number of symbols required to complete the data collection task. The bound applies to all sets of parameters on the storage-bandwidth tradeoff curve \eqref{eq:sbt}, with a caveat that for the intermediate points, we have to allow codes with functional rather than exact repair.  At the MBR point the bound 
is attainable, as shown by the above example as well as by another example that we consider in this section, namely the MBR Product-Matrix codes.

\subsection{Lower bound for the data retrieval bandwidth}\label{sec:bound_retrieval}
For non-MSR regenerating codes, the possibility of reducing the number of downloaded field symbols motivates us to seek a 
lower bound on the communication complexity. Let us formally introduce our model. Like before, an $[n,k,d,l,\beta,M]$ regenerating code is 
defined on a connected graph $G=(V,E)$. A set of $k$ nodes, denoted by $K$, wish to send their data to the DC for the purpose
of recovering the original file of size $M$. We assume DC to be an external node (if DC is a node in the graph $G$ itself then it needs to 
contact $k-1$ other nodes but the analysis remains the same.). To formalize this model, suppose that DC has direct access only to a subset $\bar{K}\subset K$ with $|\bar{K}| 
<k$ and let $G_{\bar{K},K}$ be the graph with $V_{\bar{K},K} = \{DC\}\cup K$ and $E_{\bar{K},K} = \{(DC,v): v \in \bar{K}\} \cup \{(u,v):u,v \in K, (u,v)\in E\}$. We will
assume that this graph is connected and all communication for the data retrieval process will be done on this new graph $G_{\bar{K},K}$.

\begin{lemma}\label{lemma:bound2}
	For an $[n,k,d,l,\beta,M]$ regenerating code and any set $A \subseteq K$ of size $a$, let $R_A$ be the data derived as a function $W_A$ 
	such that $H(\cF|R_A, W_{K\setminus A})=0$. Then
	\begin{equation}\label{eq:RA}
		H(R_A) \ge \sum_{i=k-a}^{k-1}\min\{l,(d-i)\beta\}.
	\end{equation}
\end{lemma}   
\begin{proof}
	From Lemma~\ref{lemma:claim1}, we know that for any set $B = \{b_1,b_2,\dots,b_{|B|}\} \subset K$, 
	$$H(W_B) = \sum_{i=1}^{|B|}H(W_{b_i}|W_{b_{i-1}},\dots,W_{b_1}) \le \sum_{i=1}^{|B|}\min \{l,(d-i+1)\beta\} = \sum_{i=0}^{|B|-1}\min \{l,(d-i)\beta\}.$$ 
	From the data retrieval property of the code, we have
	$$
		H(R_A, W_{K\setminus A}) \ge \sum_{i=0}^{k-1}\min\{l,(d-i)\beta\}, 
	$$
which implies
		\begin{align*}
		 H(R_A) &\ge \sum_{i=0}^{k-1}\min\{l,(d-i)\beta\}- H(W_{K\setminus A})\\
		&\ge \sum_{i=0}^{k-1}\min\{l,(d-i)\beta\}-\sum_{i=0}^{k-a-1}\min\{l,(d-i)\beta\}\\
		&= \sum_{i=k-a}^{k-1}\min\{l,(d-i)\beta\}.
	\end{align*}
\end{proof}

Specializing bound \eqref{eq:RA} for the MSR and MBR points, we obtain
\begin{corollary}\label{cor:MBR}
(1) For an $[n,k,d,l,\beta,M]$ MSR code and any subset $A \subseteq K$
  \begin{equation*}
  	H(W_A) \ge |A|\cdot l.
 \end{equation*}
(2)
	For an $[n,k,d,l,\beta,M]$ MBR code and any subset $A \subseteq K$  of size $a$
	\begin{equation}\label{eq:MBR}
		H(W_A) \ge \sum_{i=k-a}^{k-1}(d-i)\beta.
	\end{equation}	
\end{corollary} 
\noindent\textit{Remark:} Part (1) of this corollary gives a formal proof of our earlier claim as to why standard relaying is optimal for data retrieval with
MSR codes. 

\noindent\textit{Remark:} Note that at the MBR point, even for the fully connected setting when the DC has direct access to $k$ nodes, 
data retrieval is performed by downloading full contents of the $k$ nodes. At the same time, Cor.\ref{cor:MBR}(2) shows that it might be possible to retrieve the file by downloading fewer symbols. In the next section we show that this is indeed the case and 
that bound \eqref{eq:MBR} is achievable with PM MBR codes; thus this bound is in fact tight.

\subsection{Data retrieval with optimal communication}\label{sec:achieve_retrieval}

In this section we describe a data collection procedure on a graph with communication complexity attaining the bound \eqref{eq:MBR}, 
using the PM code family as an example. Let us first recall the standard PM MBR construction of Rashmi et al. \cite{Rashmi11}. The parameters of the codes are $[n,k,d,l=d,\beta=1,M=kd-\binom{k}{2}],$ see also \eqref{eq:MSR-MBR}. The data file $\cF$ is formed of $M$ symbols of the field $F$, and it is 
represented by a $d\times d$ matrix $B$ that has the following structure:
$$
B = \begin{bmatrix}
	S & T \\
	T^{\intercal} & 0
\end{bmatrix}.
$$
Here $S$ is a $k \times k$ symmetric matrix and $T$ is a $k \times (d-k)$ matrix. Together these two matrices contain $\binom{k}{2}+k(d-k) = dk - \binom{k}{2}$ message symbols. To encode the message, choose $n$ distinct nonzero elements $x_1,\dots,x_n$ of $F$ and use them to construct an $n\times d$ Vandermonde matrix $\Psi$ with each row formed of consecutive powers of one of the $x_i$'s. The $n \times d$ codeword matrix is 
found as
   $$
   C = \Psi B.
   $$

The data retrieval proceeds as follows. Assume that the DC aims at retrieving $\cF$ by accessing 
the stored contents of nodes $C_1,\dots,C_k$ (or any other $k$-tuple of the nodes). Denote by $\Psi_k$ the submatrix of $\Psi$ formed by the
first $k$ rows of $\Psi,$ and write it as $\Psi_k=[\Psi_{k,1}|\Psi_{k,2}],$ where $\Psi_{k,1}$ is a $k\times k$ Vandermonde matrix.
Upon retrieving the information from the nodes $C_1,\dots,C_k$, the DC 
has access to the $k\times d$  matrix
   \begin{equation}\label{eq:P}
    [
\Psi_{k,1}S+\Psi_{k,2}T^{\intercal}\mid\Psi_{k,1}T
    ],
   \end{equation}
where the left submatrix has $k$ and the right $d-k$ columns. Since $\Psi_{k,1}$ is invertible, from the right submatrix the DC can find
the matrix $T$. Once found, it gives access to the product $\Psi_{k,1}S$ and then to $S$, completing the decoding (data recovery) process.

Inspired by the polynomial description of PM MSR codes in \cite{Duursma2020}, we now present a similar description of the PM MBR codes and show how this achieves the bound \eqref{eq:MBR}. Since $B$ is a symmetric matrix, we can associate its elements with the 
coefficients of a symmetric polynomial $s(y,z)$ such that
   $$
   \min(\deg_y(s),\deg_z(s))\le k-1 \text{ and } \max(\deg_y(s),\deg_z(s))\le d-1.
   $$
Next, we let node $i$ store the $d$ coefficients of the polynomial $g_i(z) = s(x_i,z),$ where $x_i\in F$ is one of the elements chosen above. 
Altogether this forms an equivalent description of the encoding procedure of the code. 
It is clear that retrieving the coefficients of any $k$ of these polynomials results in the retrieval of the file: for instance, the coefficients
of $g_i,i=1,\dots,k$ exactly correspond to the rows of the matrix \eqref{eq:P}.

The standard data retrieval scheme described above suggests acquiring all the coefficients of $k$ polynomials. We observe that this is in fact
not necessary because the file can be recovered by accessing exactly $M$ elements of the codeword.
Without loss of generality, assume that the nodes $1,2,\dots, k$ are contacted for data retrieval. Node $i$ returns the $d-i+1$ symbols $\{g_{i}(x_{j}):j=i,i+1,\dots,d\}$. To see that this scheme achieves the bound \eqref{eq:RA},  without loss of generality let $A=\{1,2,\dots,a\}$. Then the total data transmitted by the set $A$ is $\sum_{i \in A}|\{g_{i}(x_{j}):j=i,i+1,\dots,d\}| = \sum_{i=k-a+1}^k(d-i+1)$ which matches the bound. The correctness of the scheme follows from the next lemma.
\begin{lemma}
The set of symbols $\{g_{i}(x_{j}):j=i,i+1,\dots,d,i=1,2,\dots,k\}$ are sufficient to recover the original $M$ symbols.
\end{lemma}
\begin{proof}
	We have the following set of evaluations of the symmetric polynomial $s(y,z)$ at the points $\{(x_i,x_j):j=i,i+1,\dots,d,i=1,2,\dots,k\}$. From node 1, the DC gets $d$ evaluations of the polynomial $g_1(z)$ of degree $d-1,$ sufficient to find its
coefficients. From node 2, the DC obtains $\{g_2(x_i): i=2,\dots,d\};$ since it already knows $g_2(x_1) = s(x_2,x_1) = s(x_1,x_2) = g_1(x_2),$
altogether it has therefore access to $d$ evaluations of $g_2(z)$ which again suffice to recover its coefficients. By induction, 
if the DC has recovered the coefficients of all $g_m(z)$ for $m\le i$ for some $i<k$, then after acquiring  further
$d-i$ symbols from node $i+1$, it will have access to $d$ evaluations of $g_{i+1}(z)$. This process results in
recovery of all the polynomials $\{g_i(z): 1 \le i \le k\}$, and this completes the data retrieval.
\end{proof}
\begin{example}
	{\rm Consider again the graph from Example \ref{example:pm} but this time assume that codewords of a PM MBR code with parameters
		$[n=7,k=5, d=6, l=6,\beta=1, M=20]$ are placed on the nodes. Assume that DC has direct access only to node 1, i.e., $\bar{K} = \{1\}$. 
		By the breadth-first search algorithm, the 5 nodes taking part in the data retrieval process are chosen to be nodes $1,2,3,4$ and $5$. Figure \ref{fig:pm_mbr1} shows the data transmission required in the traditional setting where each node sends its $l=6$ symbols to the DC for retrieval of the file. Applying the Corollary \ref{cor:MBR}(2), we have $H(R_{\{4\}}) \ge 2, H(R_{\{4,5\}}) \ge 5$ and so on. Figure \ref{fig:pm_mbr2} shows the optimal data transmission matching the bound \eqref{eq:MBR} under the same connectivity constraints. Examining the results, we see that optimizing the communication results in moving $27$ fewer symbols. }
	\begin{figure}[ht]	\begin{center}\scalebox{0.7}
			{
				\subfigure[Traditional data retrieval]{
					\begin{tikzpicture}[roundnode/.style={circle, draw=black},
						rootnode/.style={circle, draw=red, very thick, minimum size=1mm, fill=red}			]
						\node[roundnode] (8) at (0,2) {$DC$};
						\node[roundnode] (1) at (0,0) {1};
						\node[roundnode] (2) at (-1.5,-2) {2};
						\node[roundnode] (3) at (1.5,-2) {3};
						\node[roundnode] (4) at (-3,-4) {4};
						\node[roundnode] (5) at (-1,-4) {5};
						\node[roundnode] (6) at (1,-4) {6};
						\node[roundnode] (7) at (3,-4) {7};
						\path[<-] (1) edge [blue,thick] node[style={anchor=east}] {18} (2);
						\path[<-] (1) edge [blue,thick] node[style={anchor=east}] {6}(3);
						\path[<-] (2) edge [blue,thick] node[style={anchor=east}] {6}(4);
						\path[<-] (2) edge [blue,thick] node[style={anchor=east}] {6}(5);
						\path[-] (3) edge [thick] (6);
						\path[-] (3) edge [thick] (7);
						\path[->] (1) edge [blue,thick] node[style={anchor=east}] {30}(8);
					\end{tikzpicture}\label{fig:pm_mbr1}
				}
				\hspace{0.3in}
				\subfigure[Optimal data retrieval]{
					\begin{tikzpicture}[roundnode/.style={circle, draw=black},
						rootnode/.style={circle, draw=red, very thick, minimum size=1mm, fill=red}			]
						\node[roundnode] (8) at (0,2) {$DC$};
						\node[roundnode] (1) at (0,0) {1};
						\node[roundnode] (2) at (-1.5,-2) {2};
						\node[roundnode] (3) at (1.5,-2) {3};
						\node[roundnode] (4) at (-3,-4) {4};
						\node[roundnode] (5) at (-1,-4) {5};
						\node[roundnode] (6) at (1,-4) {6};
						\node[roundnode] (7) at (3,-4) {7};
						\path[<-] (1) edge [blue,thick] node[style={anchor=east}] {9} (2);
						\path[<-] (1) edge [blue,thick] node[style={anchor=east}] {5}(3);
						\path[<-] (2) edge [blue,thick] node[style={anchor=east}] {2}(4);
						\path[<-] (2) edge [blue,thick] node[style={anchor=east}] {3}(5);
						\path[-] (3) edge [thick] (6);
						\path[-] (3) edge [thick] (7);
						\path[->] (1) edge [blue,thick] node[style={anchor=east}] {20}(8);
					\end{tikzpicture}\label{fig:pm_mbr2}
				}
			}
		\end{center}
		\caption{Traditional vs optimal transmission for data retrieval for PM MBR code. In part (b) the graph $G_{\bar K,K}$ is formed of vertices 1 through 5, where ${\bar K}=\{1\}, K=\{2,3,4,5\}$. }\label{fig:pm_mbr}
	\end{figure}
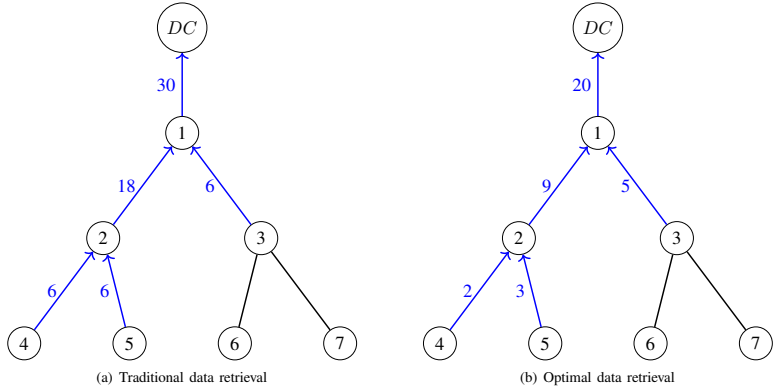
\end{example}
As noted above, even in the full connectivity setting when DC has direct access to every node, to perform data retrieval it suffices to download only $M$ symbols, not $kl>M$ symbols as proposed in the original work \cite{Rashmi11}. We believe in fact that this is a general
phenomenon that applies for all $F$-linear families of MBR codes.

%
%

\section{Error Correction during Repair}\label{sec:errors}
The tasks of node repair and data retrieval involve transmission of data across potentially unreliable links across the network. Hence the question of error control naturally arises. Additionally, there can also be adversarial nodes in the network that try to hamper the process or introduce errors in the outcomes of the process. In the repair scenario, this can lead to repaired nodes whose information is corrupted and
it can spread through the network if this faulty node takes part in further repairs as a helper. Previous works on error control during repair (\cite{RSRK2012}, \cite{Ye16a}, \cite{Silberstein2015} and others\footnote{Errors in node repair can be also framed as networks with an ``active adversary,'' and then error resilience is tied to ``secrecy capacity'' of the network \cite[Sec.9.4]{Ramkumar2022}.}) all focus on the traditional model of direct connectivity. These schemes still work in the graph scenario if the nodes rely upon standard relaying of data, i.e., the AF strategy. At the same time, if the nodes perform intermediate processing, error amplification can happen, similar to what happens in network coding with errors. This is because even a single corrupted symbol can potentially affect all the linear combinations 
evaluated at the node. Additionally, if a node is adversarial, it can also corrupt all the data it forwards to the next node in the network, even if it received correct data from all of its predecessors. 

As argued below, the case of noisy edges can be handled in a straightforward manner. If each of the edges does not introduce more than a fixed number of errors, one can simply encode each transmission using an error correcting code with sufficient minimum distance. This decoder can be installed on every node, checking the incoming information for errors and correcting them before the processing. Once the node evaluates the data passed to the next node during the repair task, it is again encoded and forwarded along the path. This local error correction
precludes error amplification during the transmission. This method faces two limitations. First, it requires the nodes to be furnished with
processing power (which is assumed to be there anyway because the nodes transmit some function of their contents), and secondly, it
assumes that the nodes themselves do not interfere with the data processing by altering it, i.e., they do not turn adversarial. 

Below we suppose that edges in the graph are noisy, assuming that each edge can introduce at most a $\rho$ proportion of errors in the data
passed over it, i.e., a vector of $N$ symbols transmitted over the edge is received by the parent node with at most $\rho N$ incorrect field symbols.

\begin{proposition}
Suppose that the error rate on the edges is bounded above by $\rho\in(0,1).$ It is possible to perform the repair process successfully using the IP technique with an additional overall transmission overhead of $(1-2\rho)^{-1}$ compared to the zero error model.
\end{proposition}

\begin{proof}
Recall the transmission scheme for PM MSR codes described in Section \ref{sec:pm}. Any set $A$ of size $|A| < d-k+1$ transmits $|A|\beta$ symbols and any set $A$ of size $|A| \ge d-k+1$ transmits $l$ symbols according to \eqref{eq:eqtn4}. To minimize the overhead, we may
encode these symbols with a Reed-Solomon code of rate $1-2\rho,$ adding a fraction of $1/(1-2\rho)$ symbols to the transmission. 
The distance of the code relative to the length is about $2\rho,$ supporting the needed error correction function at the nodes.	
\end{proof}

This proposition offers a straightforward way of handling noisy edges in the network. An observation that can be made here is that
the addition of an error-correcting code ties well with the IP processing and does not require much in the way of extra computations. 
We exemplify this remark for PM codes and the IP procedure of Sec.~\ref{sec:pm}.

Let $G_l$ be the generating matrix of an error-correcting code of dimension $l$, say a Reed-Solomon code. 
Let $|A| \ge d-k+1$ and let $h \in A$ be the vertex that transmits the encoded vector $G_l^T\xi(f,A)$ to some other node outside of $A$, 
where the vector $\xi(f,A)$ representing the noiseless communication is given in \eqref{eq:eqtn4}. Writing this product in detail, we obtain
\begin{align*}
	G_l^T{\xi}(f,A)&= \sum_{h \in A}g^{(h)}(a_f)G_l^T
	\left[\begin{array}{l}
		l^h_0+a_f^{k-1}l^h_{k-1}\\[.05in]
		l^h_1+a_f^{k-1}l^h_{k}\\[.05in]
		\hspace*{.3in}\vdots\\
		l^h_{k-2}+a_f^{k-1}l^h_{2k-3}	
	\end{array}\right] \\
	& = g^{(h)}(a_f)G_l^T
	\left[\begin{array}{l}
		l^{h}_0+a_f^{k-1}l^{h}_{k-1}\\[.05in]
		l^{h}_1+a_f^{k-1}l^{h}_{k}\\[.05in]
		\hspace*{.3in}\vdots\\
		l^{h}_{k-2}+a_f^{k-1}l^{h}_{2k-3}	
	\end{array}\right] + G_l^T{\xi}(f,A\setminus\{h\}).
\end{align*}
where the first part is the encoded contribution of the node $h$ and the second part is the encoded contribution of the rest of the nodes in 
$A$. Each of these are themselves codewords of the RS code. Observe that the encoding can be distributed between the information of the
node $h$ and the data from the other nodes. Hence, for computation purposes, once the node $h$ receives a (possibly corrupted) codeword from some other member of the set $A$, it needs only to identify the nearest codeword of what was received. 
It does not necessarily need to recover $\xi(f,A\setminus\{h\})$, but can find the next codeword to transmit simply by adding its encoded 
contribution. 

%
%
	\section{Partial Node Repair}\label{sec:partial}
The problem of partial node repair is a generalization of the traditional node repair problem. While in the traditional setting, the set of 
erased or failed nodes is assumed to have complete loss of data, in the partial repair problem it is assumed that only a part, say $\gamma 
l, 0\le\gamma\le1$, of such a node's contents are erased. One of the first works devoted to this question was \cite{Gerami2017} which 
derived a version of the bound on the file size \eqref{eq:sbt} that accounts for the parameter $\gamma$:
    $$
    M\le \sum_{i=0}^{k-1}\min\{l,(d-i)\beta+l(1-\gamma)\}.
    $$
Note that this expression gives $l=\frac{M}{k}, \beta = \frac{l\gamma}{d-k+1}$ at the MSR point.

This problem gives rise to a number of open questions, starting with MSR code constructions, that to the best of our knowledge, have not been addressed in the literature. Without attempting a comprehensive analysis, we point out that partial repair can be implemented under the IP approach discussed here. The underlying idea is that the helper nodes need to transmit only the linear combinations corresponding to the failed coordinates.

	\begin{example}
		Consider the $[n=7,k=4,d=6,l=3,\beta=1,M=12]$ PM MSR code, placed on the graph in Fig.~\ref{fig:graph}.
		 If the entire contents of the root node is lost, both the AF strategy and the IP strategy of Lemma~\ref{lemma_pm} require a total transmission of 10 symbols. At the same time, if only the first coordinate of the root node needs to be recovered, then the two immediate neighbors of the root node can transmit just the first row of equation \eqref{eq:eqtn4}, and 6 transmissions suffice.
	\end{example}

Generalizing this example, we state the following lemma whose proof is immediate.
	\begin{lemma}
		Given an $[n,k,d,l,\beta,M]$ linear regenerating code defined on a complete graph, suppose that a fraction 
$\gamma \le \frac{\beta}{l}$ of the symbols of a node are erased. To recover their values, each helper node needs to transmit only $\gamma l$ symbols to the failed node. The total communication complexity of repair is $d\gamma l$. 
	\end{lemma}

Using the notation of Sec.~\ref{sec:BRB},
let $T_f$ be a spanning tree of the repair graph $G_{f,D}$ with root at $f$. 
Let $D(v)$ be the descendants of node $v \in V_{f,D}$ and let $D^*(v) = D(v) \cup \{v\}$. We have the following lemma that generalizes our earlier result (Theorem 1 in \cite{Patra2021}).
	\begin{lemma}
Given an $[n,k,d,l,\beta,F]$ linear regenerating code. There exists a repair procedure that recovers a $\gamma$ 
fraction of the failed node, $0<\gamma<1,$ using the repair bandwidth
		$$
		\sum_{v \in D} \min\{\gamma l,|D^*(v)|\beta\}.
		$$
	\end{lemma}
\begin{proof}
Any leaf helper node transmits $\min\{\gamma l,\beta\}$ symbols to its parent in $T_f$. This is possible because if $\gamma < \frac{\beta}{l}$, then instead of the $\beta$ symbols the linear combinations of these symbols corresponding to the failed coordinates can be transmitted. Similarly any non-leaf node can transmit $\min\{|D^*(v)|\beta,\gamma l\}$. The total communication complexity is the sum of each such transmission because of the tree structure.
\end{proof}
This result entails savings in communication over the AF repair when the erased fraction of the node contents $\gamma$ is small, namely $\gamma l < \beta.$

\section{Concluding remarks}
Many facets of designing regenerating codes on networks described by general graphs still await their study. We mention two such research directions. 

While we have discussed incorporating error correction into the regeneration framework in networks with noisy edges, perhaps a more
relevant problem is that of adversarial nodes in the network. Such a node can alter the information stored in it, and provide arbitrary data for the repair and data retrieval tasks. Under full connectivity the effect of adversarial nodes on the outcome of these tasks can be
controlled by contacting larger groups of nodes: for instance if the code is designed to perform repair using the data from $d$ helpers,
and at most $t$ nodes in the network can be adversarial, then contacting $d+2t$ helpers and using specially designed regenerating codes
still supports reliable storage and repair \cite{RSRK2012,Ye16a}. At the same time, for regenerating codes on graphs, performing IP repair or
data retrieval in the presence of adversarial nodes looks difficult because even one such node on the path in the graph from the helpers to the 
failed node can corrupt the aggregated information entering it from a large group of nodes. A solution utilized in network coding for graphs 
with noisy links relies on rank metric codes \cite{KK2008}, however in the node repair paradigm it is necessary to combine them with 
regenerating codes to control the communication complexity as well as the spread of the errors, which does not look immediate. The difficulty
of handling adversarial nodes has been recognized in the network coding community, where the only paper that addressed it, \cite{Che2013}, limited its scope to optimized routing, stopping short of designing code constructions.

Another way of accounting for the distance to the helper nodes potentially looks outside the domain of regenerating codes for bandwidth saving. Namely, it could be possible to design a coding system that, once faced with a repair task, identifies a helper set of graph vertices and
proceeds with downloading the amount of information inverse proportional to the distance to the failed node. Such a code should be able 
to handle any choice of the failed node and the helper set in a uniform fashion. While most known constructions of regenerating codes
rely on uniform download (which is also necessary and sufficient for attaining the cutset bound \eqref{eq:MSR-MBR} for the MSR case), there
are some works on heterogeneous storage systems based on different volumes of the data passed from different subsets of helper nodes, see e.g.,
\cite{AkhKiaGha2010cost}. At the same time, none of these schemes allow for a flexible choice of those subsets, which therefore remains an open
problem.

\appendix
\subsection{Proof of Lemma \ref{lemma:cA1} {\rm (see \cite{Elyasi2020}; \cite[p.631-2]{Ramkumar2022})}}
	Let $A_{\sim i}= A\setminus\{i\}, A_y = A\cup \{y\}, A_{\sim i,y} = (A_y)_{\sim i}$. Below we denote the elements of the matrices $R,D$ and $\Phi$ by lowercase letters. Recall also that $w_{S,i}$ denote elements of the set $\cW$.
	\begin{align*}
\sum_{i \in A}(-1)^{\tau_A(i)}R_{A\setminus\{i\},:}D_{:,i}
&=\sum_{i \in A}(-1)^{\tau_A(i)}\sum_{L\subset [d], |L|=m}r_{A\sim i, L}d_{L,i}\\
&= \sum_{i \in A}(-1)^{\tau_A(i)}r_{A\sim i, A}d_{A,i} + \sum_{i \in A}(-1)^{\tau_A(i)}\sum_{y\in [d]\setminus A}r_{A_{\sim i},A_{\sim i,y}}d_{A_{\sim i,y},i}\\
&= \sum_{i\in A}\phi_{i,1}d_{A,i}+\sum_{i \in A}(-1)^{\tau_A(i)}\sum_{y\in [d]\setminus A}(-1)^{\tau_{A_{\sim i,y}}(y)}\phi_{y,1}d_{A_{\sim i,y},i}\\
& = \sum_{i\in A}\phi_{i,1}d_{A,i}+\sum_{y\in [d]\setminus A}\phi_{y,1}\sum_{i \in A}(-1)^{\tau_A(i)+\tau_{A_{\sim i,y}}(y)}w_{A_y,i}.
	\end{align*} 
Now for $i \ne y$,
\begin{align*}
\tau_A(i)&+\tau_{A_{\sim i,y}}(y) = |\{j \in A: j \le i\}|+|\{x \in A_{\sim i,y}:x \le y\}|\\
&= |\{j \in A_y: j \le i\}-\mathbbm{1}(y <i)+|\{x \in A_y:x \le y\}|- \mathbbm{1}(i <y)\\
&= |\{j \in A_y: j \le i\}+|\{x \in A_y:x \le y\}|-1\\
&= \tau_{A_y}(i)+\tau_{A_y}(y)-1,
\end{align*}
and we obtain
\begin{align*}
\sum_{i \in A}(-1)^{\tau_A(i)}R_{A\setminus\{i\},:}D_{:,i}  &= \sum_{i\in A}\phi_{i,1}d_{A,i}+\sum_{y\in [d]\setminus A}\phi_{y,1}\sum_{i \in A}(-1)^{\tau_{A_y}(i)+\tau_{A_y}(y)-1}w_{A_y,i}\\ 
&= \sum_{i\in A}\phi_{i,1}d_{A,i}+\sum_{y\in [d]\setminus A}(-1)^{\tau_{A_y}(y)}\phi_{y,1}\sum_{i \in A}[-(-1)^{\tau_{A_y}(i)}w_{A_y,i}]\\
&=\sum_{i\in A}\phi_{i,1}d_{A,i}+\sum_{y\in [d]\setminus A}(-1)^{\tau_{A_y}(y)}\phi_{y,1}(-1)^{\tau_{A_y}(y)}w_{A_y,y}\\
&= \sum_{i\in A}\phi_{i,1}d_{A,i} +\sum_{y \in [d]\setminus A}\phi_{y,1}d_{A,y} = \sum_{i\in [d]}d_{A,i}\phi_{i,1} = c_{A,1}.
\end{align*}

\bibliographystyle{IEEEtranS}
\bibliography{references}	

\begin{thebibliography}{10}
\providecommand{\url}[1]{#1}
\csname url@samestyle\endcsname
\providecommand{\newblock}{\relax}
\providecommand{\bibinfo}[2]{#2}
\providecommand{\BIBentrySTDinterwordspacing}{\spaceskip=0pt\relax}
\providecommand{\BIBentryALTinterwordstretchfactor}{4}
\providecommand{\BIBentryALTinterwordspacing}{\spaceskip=\fontdimen2\font plus
\BIBentryALTinterwordstretchfactor\fontdimen3\font minus
  \fontdimen4\font\relax}
\providecommand{\BIBforeignlanguage}[2]{{%
\expandafter\ifx\csname l@#1\endcsname\relax
\typeout{** WARNING: IEEEtranS.bst: No hyphenation pattern has been}%
\typeout{** loaded for the language `#1'. Using the pattern for}%
\typeout{** the default language instead.}%
\else
\language=\csname l@#1\endcsname
\fi
#2}}
\providecommand{\BIBdecl}{\relax}
\BIBdecl

\bibitem{AkhKiaGha2010cost}
S.~Akhlaghi, A.~Kiani, and M.~R. Ghanavati, ``Cost-bandwidth tradeoff in
  distributed storage systems,'' \emph{Computer Communications}, vol.~33,
  no.~17, pp. 2105--2115, 2010.

\bibitem{Che2013}
P.~H. Che, M.~Chen, T.~Ho, S.~Jaggi, and M.~Langberg, ``Routing for security in
  networks with adversarial nodes,'' in \emph{2013 International Symposium on
  Network Coding (NetCod)}, 2013, pp. 1--6.

\bibitem{Dimakis10}
A.~G. {Dimakis}, P.~B. {Godfrey}, Y.~{Wu}, M.~J. {Wainwright}, and
  K.~{Ramchandran}, ``Network coding for distributed storage systems,''
  \emph{IEEE Trans. Inf. Theory}, vol.~56, no.~9, pp. 4539--4551, 2010.

\bibitem{Duursma2021}
I.~M. Duursma, X.~Li, and H.-P. Wang, ``Multilinear algebra for distributed
  storage,'' \emph{SIAM J. Appl. Algebra Geom.}, vol.~5, pp. 552--587, 2021.

\bibitem{Duursma2020}
\BIBentryALTinterwordspacing
I.~M. Duursma and H.-P. Wang, ``Multilinear algebra for minimum storage
  regenerating codes: a generalization of the product-matrix construction,''
  \emph{Applicable Algebra in Engineering, Communication and Computing}, 2021.
  [Online]. Available: \url{https://doi.org/10.1007/s00200-021-00526-3}
\BIBentrySTDinterwordspacing

\bibitem{Elyasi2020}
M.~{Elyasi} and S.~{Mohajer}, ``Cascade codes for distributed storage
  systems,'' \emph{IEEE Trans. Inf. Theory}, vol.~66, no.~12, pp. 7490--7527,
  2020.

\bibitem{Elyasi2016}
M.~Elyasi and S.~Mohajer, ``Determinant coding: {A} novel framework for
  exact-repair regenerating codes,'' \emph{IEEE Trans. Inf. Theory}, vol.~62,
  no.~12, pp. 6683--6697, 2016.

\bibitem{Elyasi2019}
------, ``Determinant codes with helper-independent repair for single and
  multiple failures,'' \emph{IEEE Trans. Inf. Theory}, vol.~65, no.~9, pp.
  5469--5483, 2019.

\bibitem{Gerami2017}
M.~Gerami, M.~Xiao, and M.~Skoglund, ``Two-layer coding in distributed storage
  systems with partial node failure/repair,'' \emph{IEEE Communications
  Letters}, vol.~21, no.~4, pp. 726--729, 2017.

\bibitem{GeramiXiao2014}
M.~Gerami and M.~Xiao, ``Exact optimized-cost repair in multi-hop distributed
  storage networks,'' in \emph{2014 IEEE International Conference on
  Communications (ICC)}, 2014, pp. 4120--4124.

\bibitem{KK2008}
R.~Koetter and F.~R. Kschischang, ``Coding for errors and erasures in random
  network coding,'' \emph{IEEE Transactions on Information Theory}, vol.~54,
  no.~8, pp. 3579--3591, 2008.

\bibitem{LiMowDengWu2022}
Z.~Li, W.~H. Mow, L.~Deng, and T.-Y. Wu, ``Optimal-repair-cost {MDS} array
  codes for a class of heterogeneous distributed storage systems,'' in
  \emph{2022 IEEE International Symposium on Information Theory (ISIT)}, 2022,
  pp. 2379--2384.

\bibitem{LuXuanFu2014}
J.~Lu, X.~Guang, and F.-W. Fu, ``Distributed storage over unidirectional ring
  networks,'' in \emph{2014 International Symposium on Information Theory and
  its Applications}, 2014, pp. 368--372.

\bibitem{Mohajer2015}
S.~Mohajer and R.~Tandon, ``New bounds on the $(n, k, d)$ storage systems with
  exact repair,'' in \emph{2015 IEEE International Symposium on Information
  Theory (ISIT)}, 2015, pp. 2056--2060.

\bibitem{Patra2021}
A.~Patra and A.~Barg, ``Node repair on connected graphs,'' \emph{IEEE Trans.
  Inf. Theory}, vol.~68, no.~5, pp. 3081--3095, 2022.

\bibitem{Ramkumar2022}
V.~Ramkumar, S.~Balaji, B.~Sasidharan, M.~Vajha, M.~N. Krishnan, and P.~V.
  Kumar, ``Codes for distributed storage,'' \emph{Foundations and Trends in
  Communications and Information Theory}, vol.~19, pp. 547--813, 2022.

\bibitem{Rashmi11}
K.~V. {Rashmi}, N.~B. {Shah}, and P.~V. {Kumar}, ``Optimal exact-regenerating
  codes for distributed storage at the {MSR} and {MBR} points via a
  product-matrix construction,'' \emph{IEEE Trans. Inf. Theory}, vol.~57,
  no.~8, pp. 5227--5239, 2011.

\bibitem{RSRK2012}
K.~V. Rashmi, N.~B. Shah, K.~Ramchandran, and P.~V. Kumar, ``Regenerating codes
  for errors and erasures in distributed storage,'' in \emph{Proc. IEEE
  International Symposium on Information Theory, Cambridge, MA, USA}, 2012, pp.
  1202--1206.

\bibitem{Rotman2009}
J.~Rotman, \emph{An Introduction to Homological Algebra}, 2nd~ed.\hskip 1em
  plus 0.5em minus 0.4em\relax New York, N.Y.: Springer, 2009.

\bibitem{Sasidharan2014}
B.~Sasidharan, K.~Senthoor, and P.~V. Kumar, ``An improved outer bound on the
  storage-repair-bandwidth tradeoff of exact-repair regenerating codes,'' in
  \emph{2014 IEEE International Symposium on Information Theory}, 2014, pp.
  2430--2434.

\bibitem{Senthoor2015}
K.~Senthoor, B.~Sasidharan, and P.~V. Kumar, ``Improved layered regenerating
  codes characterizing the exact-repair storage-repair bandwidth tradeoff for
  certain parameter sets,'' in \emph{2015 IEEE Information Theory Workshop
  (ITW)}, 2015, pp. 1--5.

\bibitem{Shah2012}
N.~B. {Shah}, K.~V. {Rashmi}, P.~V. {Kumar}, and K.~{Ramchandran},
  ``Distributed storage codes with repair-by-transfer and nonachievability of
  interior points on the storage-bandwidth tradeoff,'' \emph{IEEE Trans. Inf.
  Theory}, vol.~58, no.~3, pp. 1837--1852, 2012.

\bibitem{Silberstein2015}
N.~Silberstein, A.~S. Rawat, and S.~Vishwanath, ``Error-correcting regenerating
  and locally repairable codes via rank-metric codes,'' \emph{IEEE Trans.
  Inform. Theory}, vol.~61, no.~11, pp. 5765--5778, 2015.

\bibitem{SohnChoYooMoo2018}
J.~Y. Sohn, B.~Choi, S.~W. Yoon, and J.~Moon, ``Capacity of clustered
  distributed storage,'' \emph{IEEE Trans. Inf. Theory}, vol.~65, no.~1, pp.
  81--107, 2019.

\bibitem{Tian14}
C.~Tian, ``Characterizing the rate region of the $(4, 3, 3)$ exact-repair
  regenerating codes,'' \emph{{IEEE} J. Sel. Areas Commun.}, vol.~32, no.~5,
  pp. 967--975, 2014.

\bibitem{Ye16a}
M.~Ye and A.~Barg, ``Explicit constructions of optimal-access {MDS} codes with
  nearly optimal sub-packetization,'' \emph{IEEE Trans. Inf. Theory}, vol.~63,
  no.~10, pp. 6307--6317, 2017.

\bibitem{Yeung2006}
R.~W. Yeung, S.-Y.~R. Li, N.~Cai, and Z.~Zhang, ``Network coding theory {P}art
  {I}: {S}ingle source,'' \emph{Foundations and Trends® in Communications and
  Information Theory}, vol.~2, no.~4, pp. 241--329, 2006.

\end{thebibliography}
\end{document}